\documentclass[twocolumn]{aastex63}

\usepackage{hyperref}

\newcommand{\nx}{N_{\rm I}} 
\newcommand{\nt}{N_{\rm T}} 
\newcommand{\nw}{N_{\rm W}} 

\usepackage{mathtools}
\usepackage{bm}
\usepackage{amsthm}
\newtheorem{theorem}{Theorem}
\newtheorem{corollary}{Corollary}
\newtheorem{lemma}{Lemma}
\newtheorem{proposition}{Proposition}[section]
\newtheorem{assumption}{Assumption}

\newcommand{\rev}[1]{{#1}}

\begin{document}
\title{Change point detection and image segmentation for time series of astrophysical images}

\shorttitle{change point detection and image segmentation}

\author[0000-0003-2762-6057]{Cong Xu}
\affil{Department of Statistics, University of California, Davis, One Shields Avenue, Davis, CA 95616, USA}
\author[0000-0003-4243-2840]{Hans Moritz G\"unther}
\affil{MIT, Kavli Institute for Astrophysics and Space Research, 77 Massachusetts Avenue, Cambridge, MA 02139, USA}
\author[0000-0002-3869-7996]{Vinay L.\ Kashyap}
\affil{Center for Astrophysics $|$ Harvard \& Smithsonian, 60 Garden Street, Cambridge, MA 02138, USA}
\author[0000-0001-7067-405X]{Thomas C. M. Lee}
\affil{Department of Statistics, University of California, Davis, One Shields Avenue, Davis, CA 95616, USA}
\author{Andreas Zezas}
\affil{Department of Physics, University of Crete, 710 03 Heraklion, Crete, Greece}


\begin{abstract}
Many astrophysical phenomena are time-varying, in the sense that their intensity, energy spectrum, and/or the spatial distribution of the emission suddenly change. This paper develops a method for modeling a time series of images. Under the assumption that the arrival times of the photons follow a Poisson process, the data are binned into 4D grids of voxels (time, energy band, and x-y coordinates), and viewed as a time series of non-homogeneous Poisson images. The method assumes that at each time point, the corresponding multi-band image stack is an unknown 3D piecewise constant function including Poisson noise. It also assumes that all image stacks between any two adjacent change points (in time domain) share the same unknown piecewise constant function. The proposed method is designed to estimate the number and the locations of all the change points (in time domain), as well as all the unknown piecewise constant functions between any pairs of the change points. The method applies the minimum description length (MDL) principle to perform this task. A practical algorithm is also developed to solve the corresponding complicated optimization problem. Simulation experiments and applications to real datasets show that the proposed method enjoys very promising empirical properties. Applications to two real datasets, the XMM observation of a flaring star and an emerging solar coronal loop, illustrate the usage of the proposed method and the scientific insight gained from it.
\end{abstract}


\NewPageAfterKeywords

\section{Introduction}\label{s:intro}

Many phenomena in the high-energy universe are time-variable, from coronal flares on the smallest stars to accretion events in the most massive black holes. Often, this variability can just be seen \rev{``}by-eye" but at other times, we need to use robust methods founded in statistics to distinguish random noise from significant variability.
Realizing where the change has occurred is critical for subsequent scientific analyses, e.g., spectral fitting and light curve modeling. Such analyses must focus on those intervals in data space which are properly tied to the changes in the physical processes that generate the observed photons. Therefore, it is of importance to identify sources as well as to locate their spatial boundaries.
Our goal is to detect change points in the time direction; that is, the times at which sudden changes happened during the underlying astrophysical process.

Change point detection in time series is well studied, and several algorithms employing different philosophies have been developed.
For example, \citet{doi:10.1111/j.1467-9892.2012.00819.x} employed hypothesis testing to study structural break detection, using both non-parametric approaches like cumulative sum (CUSUM) and parametric methods like likelihood ratio statistic to deal with different kinds of structural breaks.
Another likelihood-based approach commonly used to analyze astronomical time series is Bayesian Blocks \citep{Scargle_2013} which finds change points by fitting piecewise constant models between change points.
A good example of \rev{the} model driven approach is the Auto-PARM procedure developed by \citet{davis-et-al-2006}. By modeling the piecewise-stationary time series, the procedure is able to simultaneously estimate the number of change points, their locations and the parametric model for each piece. Here the minimum description length (MDL) principle by \citet{10.5555/534247, Rissanen2007} is applied in the model selection procedure. \citet{davis_yau_2013} proved the strong consistency of the MDL-based change point detection procedure. Another example is Automark by \citet{wong2016}, who developed {an MDL-based} methodology which detects the changes in observed emission from astronomical sources in 2-D time-wavelength space.

\citet{dey-et-al-2010} loosely classified image segmentation techniques into two categories: (i) image driven approaches and (ii) model driven approaches. Image driven segmentation techniques are mainly based on the discrete pixel values of the image. 
For example, the graph-based algorithm by \citet{felzenszwalb-2004} treats pixels as vertices, and the weights of the edges are based on the similarity between the features of pixels. The evidence for a boundary between two regions can be measured based on the graph. Such methods work in many complicated cases as there is no underlying model for images. 
Model driven approaches rely upon the information of the structure of the image. These methods are based on the assumption that the pixels in the same region have similar characteristics. The Blobworld framework by \citet{carson-et-al-1999} assumes that the features of pixels are from a underlying multivariate Gaussian mixture model. Neighboring pixels whose features are from the same Gaussian distribution are grouped into the same region. 

Here we follow the model driven approach. 
In order to develop a change point detection method for image time series data, we begin by specifying an underlying statistical model for the images between any two consecutive change points. In doing so we also study the statistical properties of the change point detection method.
We assume the underlying Poisson rate for each of the images follows a piecewise constant function. Therefore, the region growing algorithm developed by \citet{adams-bischof-1994} for greyvalue images can be naturally applied.

Given the previous successes of applying the MDL principle \citep{10.5555/534247, Rissanen2007} to other time series change point detection and image segmentation problems \citep[e.g.,][]{davis-et-al-2006,lee-2000,wong2016}, here we also use MDL to tackle our problem of joint change point detection and image segmentation for time series of astronomical images. Briefly, MDL defines the best model as the one that produces the best \rev{lossless} compression of the data. There are different versions of MDL, and the one we use is the so-called two-part code; a gentle introduction can be found in \citet{Lee01:MDLtut}. 
\rev{When comparing with other versions of MDL such as normalized maximum likelihood, one advantage of the two-part version is that it tends to be more computationally tractable for complex problems such as the one this paper considers.  It has also been shown to enjoy excellent theoretical and empirical properties in other model selection tasks~\citep[e.g.,][]{Aue-Lee11, lee-2000, davis-et-al-2006, davis_yau_2013}} Based on MDL, we develop a practical algorithm that can be applied to simultaneously estimate the number and locations of the change points, as well as to perform image segmentation on each of the images.

\section{Methodology}
Our method is applied to 4-D data cubes where 2-D spatial slices in several energy passbands are stacked in time.  Such data cubes are commonly available, though in high-energy astrophysics, data are usually obtained in the form of a list of photons.  The list contains the two-dimensional spatial coordinates where the photons were recorded on the detector, the times they were recorded, and their energies or wavelengths.
To facilitate our analysis, we bin these data into a 4-D rectangular grid of boxes. After the binning of the original data, we obtain a 4D table of photon counts indexed by the two-dimensional coordinates $(x, y)$, time index $t$ and energy band $w$.
The dataset is thus a series of multi-band images with counts of photons as the values of the pixels.
Since the emission times of photons can be considered a non-homogeneous Poisson process, and the grids do not overlap with each other, the counts in each pixel are independent, and the image slices are also independent.

We first partition these images into a set of non-overlapping region segments using a seeded region growing (SRG) method, and then merging adjacent segments to minimize MDL (see Section~\ref{subsection:SRG}).  The counts in each segment are modeled as Poisson counts (see Section~\ref{subsection:modeling}; the implementation details of the algorithm are described in Section~\ref{subsection:algorithm}).  We minimize the MDL criterion across the images by iteratively removing change points along the time axis and applying the SRG segmentation onto the images in each of the time \rev{intervals}.  Key pixels that are influential in how the segmentations and change points are determined are then identified through searching for changes in the fitted intensities (see Section~\ref{subsection:highlight}).  Such regions are the focus of follow-up analyses.

We list the variables, parameters, and notation used here in Table~\ref{tab:notations}.

\begin{deluxetable*}{cl}
\tablenum{1}
\tablecaption{Statistical Notations \label{tab:notations}}
\tablewidth{0pt}
\tablehead{
\colhead{Notation} & \colhead{Definition}
}
\startdata
$\nx$ & number of pixels in each 2-D spatial image\\
$\nt$ & number of time bins\\
$\nw$ & number of energy bands\\
$\Delta T_{t}$ & duration of the $t^{\rm th}$ time bin\\
$y_{i,t,w}$ & photon counts within the $i^{\rm th}$ spatial pixel, the $t^{\rm th}$ time interval and the $w^{\rm th}$ energy range\\
$\lambda_{i,t,w}$ & Poisson rate for the $i^{\rm th}$ spatial pixel, the $t^{\rm th}$ time interval and the $w^{\rm th}$ energy range\\
$K$ & number of change points\\
$\tau_{k}$ & location of the $k^{\rm th}$ change point\\
$m^{(k)}$ & number of region segments for the $k^{\rm th}$ interval between two consecutive change points\\
$a_{h}^{(k)}$ & the area (number of pixels) of the  $h^{\rm th}$ region segment of the $k^{\rm th}$ interval between two consecutive change points\\
$b_{h}^{(k)}$ & the ``perimeter'' (number of pixel edges between this and neighboring regions) of the $h^{\rm th}$ region of the $k^{\rm th}$ interval\\
$\mu_{h,w}^{(k)}$ & Poisson rate for the $h^{\rm th}$ region segment and the $w^{\rm th}$ energy range of the $k^{\rm th}$ interval\\
$\hat{\mu}_{h,w}^{(k)}$ & fitted Poisson rate for the $h^{\rm th}$ region segment and the $w^{\rm th}$ energy range of the $k^{\rm th}$ interval\\
\enddata
\end{deluxetable*}

\subsection{Region Growing and Merging} 
\label{subsection:SRG}
As a first step in the analysis, a suitable segmentation method must be applied to the images to delineate regions of interest (ROIs).  For this, we use the seeded region growing (SRG) method of \citet{adams-bischof-1994} to obtain a segmentation of the image. \rev{We choose SRG over other image segmentation algorithms for its speed and reliability~\citep{fan2015}.  Also, it can be straightforwardly incorporated to the Poisson setting.}

At the beginning of SRG, we select a set of seeds, manually or automatically, from the image.  Each seed can be a single pixel or a set of connected pixels. A seed comprises an initial region. Then each region starts to grow outward until the whole image is covered. \rev{(See Section~\ref{sec:practical} offers some suggestions on the selection of initial seeds.)} 

At each step, the unlabelled pixels which are neighbors to at least one of the current regions comprise the set of candidates for growing the region. 
One of these candidates is selected to merge into the region, based on the Poisson likelihood that measures the similarity between a candidate pixel and the corresponding region.  We repeat this process until all the pixels are labeled, thus producing an initial segmentation by SRG.

At the end of the SRG process, we are left with an oversegmentation, i.e., with the image split into a larger than optimal number of segments.  We then merge these segments based on the largest reduction or smallest increase in the MDL criterion (see below).  From this sequence of segmentations, we select the one that gives the smallest value of the MDL criterion as the final ROIs.

\subsection{Modeling a Poisson Image Series}
\label{subsection:modeling}
\subsubsection{Input Data Type}

We require that the data are binned into photon counts in an $\nx \times \nt \times \nw$ tensor $\{y_{i,t,w}\},i=1,...,\nx,t=1,...,\nt,w=1,...,\nw$, where $y_{i,t,w}$ is the photon counts within the $i^{\rm th}$ spatial rectangular region, the $t^{\rm th}$ time interval $[T_{t-1}, T_t)$ and the $w^{\rm th}$ energy range $[W_{w-1}, W_w)$.
After binning, the data can be viewed as a time series of images in different energy bands. 
The values of each pixel are the photon counts in the different bands in the corresponding spatial region.

Notice that compared with Automark \citep{wong2016}, we incorporate 2-D spatial information into the model, thus extending the analysis from two  (wavelength/energy and time) to four dimensions (wavelength/energy, time, and projected sky location).  We also relax the restriction that the bin sizes along any of the axes are held fixed.  Thus, sharp changes are more easily detected.

As the data in high-energy astrophysics are photon counts, we use a Poisson process to model the data,
\begin{eqnarray}
\label{equation:poisson}
    y_{i,t,w} \stackrel{i.i.d.}{\sim} \text{Poisson}(\lambda_{i,t,w} \Delta T_{t}),
\end{eqnarray}
where $\Delta T_{t} = (T_t-T_{t-1})$.

Our goal is to infer model intensities $\lambda_{i,t,w}$ from the observed counts data $\{y_{i,t,w}\}$. We are especially interested in detecting significant changes of $\lambda_{i,t,w}$ over time. If there are changes, we also want to estimate the number and locations of the change points.

To simplify the presentation, we first develop a time-homogeneous model,
i.e., one where there are no change points and $\lambda_{i,t,w}$ is unchanging with $t$ (Section~\ref{subsub:const}).  We will then consider more complex cases, where change points are added to the model so that $\lambda_{i,t,w}$ is allowed to change over time (Section~\ref{subsub:change})

\subsubsection{Piecewise Constant Model}\label{subsub:const}

First consider a temporally homogeneous Poisson model without any change points.  Then each image can be treated as an independent Poisson realization of the same, unknown, true image.

We model the image as a 3-dimensional piecewise constant function. That is, the 2-dimensional space of $x$-$y$ coordinates is partitioned into $m$ non-overlapping regions such that all the pixels in a given region have the same Poisson intensity.
Different energy bands share the same spatial partitioning. Rigorously, the Poisson parameter $\lambda_{i,t,w}$ can be written as a summation of region-specific Poisson rates $\mu_{h,w}$ times the corresponding indicator functions of regions ($I_{\{ i \in R_{h}\}}$ is 1 for pixel $i$ in region $R_{h}$ and 0 otherwise) in the following format:
\begin{eqnarray}
    \lambda_{i,t,w} =  \sum_{h=1}^{m}  \mu_{h,w} I_{\{ i \in R_{h}\}}.
\end{eqnarray}
Here $i \in R_{h}$ means \rev{``}the $i^{\rm th}$ pixel in the $h^{\rm th}$ region" and $I$ is the indicator function. $R_{h}$ is the index set of the pixels within the $h^{\rm th}$ region, with $R_{h} \subseteq \{1,...,\nx \}$. Also, $\mu_{h,w}$ is the Poisson rate for the $w^{\rm th}$ band of the $h^{\rm th}$ region. The partition of the image is specified by $\bm{R} = \{R_{h} | h=1,...,m\}$.

\subsubsection{Adding Change Points to the Model}\label{subsub:change}

Now we allow the underlying Poisson parameter $\lambda_{i,t,w}$ to change over time $t$. We model $\lambda_{i,t,w}$ as a piecewise constant function of $t$.

Suppose these $\nt$ images can be partitioned into $K+1$ homogeneous intervals by $K$ change points
$$\mathbf{\tau} = \{ \tau_0=0,~ \tau_1,\tau_2,...,\tau_K, ~\tau_{K+1}=\nt \}$$
For the $t^{\rm th}$ image, suppose that it belongs to the $k^{\rm th}$ time interval; i.e., $t \in (\tau_{k-1},\tau_{k}]$. For each given $t$, let $\lambda$ be a two-dimensional piecewise constant function with $m^{(k)}$ constant regions. Then $\lambda$ can be represented by
\begin{eqnarray}
    \lambda_{i,t,w} =  \sum_{k=1}^{K+1} I_{ \{ t \in (\tau_{k-1},\tau_{k}] \} } \sum_{h=1}^{m^{(k)}} \mu_{h,w}^{(k)} I_{\{ i \in R_{h}^{(k)}\}},
\end{eqnarray}
where $m^{(k)}$ is the number of regions within the $k^{\rm th}$ interval. Let $\mathcal{M} = \{m^{(k)} | k=1,2,...,K+1 \}$. The partition of the images within interval $k$ is specified by $\bm{R}^{(k)} = \{R_{h}^{(k)} | h=1,...,m^{(k)}\}$. And the overall partition is $\mathcal{R} = \{\bm{R}^{(k)} | k=1,2,...,K+1\}$. The Poisson rates $\mu_{h,w}^{(k)}$ is the value for the $w^{\rm th}$ band in the $h^{\rm th}$ region of the $k^{\rm th}$ interval. Let $\bm{\mu}^{(k)} = \{\mu_{h,w}^{(k)} | h=1,...,m^{(k)}, w = 1,...,\nw \}$. And let $\bm{\mu} = \{\bm{\mu}^{(k)} | k=1,...,K+1\}$, and $i \in R_{h}^{(k)}$ means \rev{``}the $i^{\rm th}$ pixel is in the $h^{\rm th}$ region of the $k^{\rm th}$ interval".

\subsubsection{Model Selection Using MDL}

Given the observed images $\{y_{i,t,w}\}$, we aim to obtain an estimate of $\lambda_{i,t,w}$. In other words, we want an estimate of the image partitions and the Poisson rates of the regions for each band. 
It is straightforward to estimate the Poisson intensities given the region partitioning, but the partitioning is a much more complicated model selection problem.

We will apply MDL to select the best fitting model. Loosely speaking, the idea behind MDL for model selection is to first obtain a MDL criterion for each possible model, and then define the best fitting model as the minimizer of this criterion. MDL defines the best model as the one that produces the best compression of the data. The criterion can be treated as the code length, or amount of hardware memory required to store the data.

First we present the MDL criterion for the homogeneous Poisson model, then follow it by the MDL criterion for the general case (i.e., with change points).

Following similar arguments as in \citet{lee-2000} (see their Appendix~B), the MDL criterion for segmenting $\nt$ homogeneous images is 
\begin{eqnarray}
    \text{MDL}(m, R, \hat{\mu}) = m \log(\nx) + \frac{\log(3)}{2} \sum_{h=1}^m b_{h} + \nonumber \\
    \frac{\nw}{2} \sum_{h=1}^m \log(\nt a_h) - \sum_{w=1}^{\nw}  \sum_{t=1}^{\nt} \sum_{h=1}^m \sum_{i \in R_{h}} y_{i,t,w} \log(\hat{\mu}_{h,w}), \nonumber \\
\end{eqnarray}
where $a_{h}$ and $b_{h}$ are, respectively, the \rev{``}area" (number of pixels) and \rev{``}perimeter" (number of pixel edges) of region $R_{h}$, and 
\begin{equation}
\hat{\mu}_{h, w} = \frac{1}{\sum_{t=1}^{\nt} \Delta T_{t} a_{h}} \sum_{t=1}^{\nt} \sum_{i \in R_{h}} y_{i,t,w}
\end{equation}
is the maximum likelihood estimate of the Poisson rate in the corresponding region. Note that the indices of $\hat{\bm{\mu}}=\{\hat{\mu}_{hw}\}$ run over the region segments $h=1..m$ and the passbands $w=1..\nw$.

For the Poisson model with change points, once the number of change points $K$ and the locations $\bm{\tau} = \{ \tau_1,...,\tau_K \}$ are specified, for each $k \in (1,2,...,K+1)$, $m^{(k)}$ and $\bm{R}^{(k)}$ can be estimated independently. Using the previous argument, the MDL criterion for images within the same homogeneous interval is
\begin{eqnarray}
    &&\text{MDL}(\tau_{k-1}, \tau_k, m^{(k)},\bm{R}^{(k)}, \hat{\bm{\mu}}^{(k)}) \nonumber \\
    &=& m^{(k)} \log(\nx) + \frac{\log(3)}{2} \sum_{h=1}^{m^{(k)}} b_{h}^{(k)} \nonumber \\
    &&+ \frac{\nw}{2} \sum_{h=1}^{m^{(k)}} \log((\tau_k - \tau_{k-1}) a_h^{(k)}) \nonumber \\
    &&- \sum_{w=1}^{\nw}  \sum_{t = \tau_{k-1}+1}^{\tau_k} \sum_{h=1}^{m^{(k)}} \sum_{i \in R_{h}^{(k)}} y_{i,t,w} \log(\hat{\mu}_{h,w}^{(k)}).
\end{eqnarray}
Then the overall MDL criterion for the model with change points is
\begin{eqnarray}
    &&\text{MDL}_{\text{overall}}(K,\mathbf{\tau},\mathcal{M},\mathcal{R},\hat{\bm{\mu}}) \nonumber \\
    = && K \log(\nt) + \sum_{k=1}^{K+1} \text{MDL}(\tau_{k-1}, \tau_k, m^{(k)},\bm{R}^{(k)}, \hat{\bm{\mu}}^{(k)}). \nonumber \\
\label{eqn:MDL}
\end{eqnarray}
To sum up, using the MDL principle, the best-fit model is defined as the minimizer of the criterion~(\ref{eqn:MDL}). The next subsection presents a practical algorithm for carrying out this minimization.

\subsubsection{Statistical Consistency}\label{sec:consistency}

An important step to demonstrating the efficacy of our method is to establish its statistical consistency.  That is, if it is shown that as the size of the data increases, the differences between the estimated model parameters and the true values decrease to zero, then the method can be said to be free of asymptotic bias, can be applied in the general case, and is elevated above a heuristic. We prove in Appendix~\ref{appendix:A} that the MDL-based model selection to choose the region partitioning, as well as the corresponding Poisson intensity parameters, is indeed strongly statistically consistent under mild assumptions of maintaining the temporal variability structure of $\lambda_{i,t,w}$.

\subsection{Practical Minimization}
\label{subsection:algorithm}

\subsubsection{An Iterative Algorithm}

Given its complicated structure, global minimization of $\text{MDL}_{\text{overall}}(K,\tau,\mathcal{M},\mathcal{R},\hat{\mu})$ (Equation~\ref{eqn:MDL}) is virtually infeasible when the number of images $N_{T}$ and the number of pixels $N_{I}$ are not small, because the time complexity of the exhaustive search is of order $2^{\nt \nx}$.

\begin{figure}[ht!]
\plotone{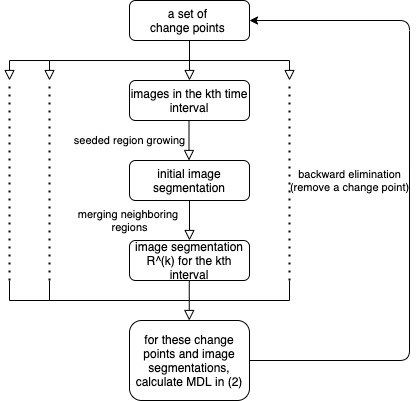}
\caption{Schematic illustration of the minimization algorithm}
\label{fig:flowchart}
\end{figure}

We iterate the following two steps to (approximately) minimize the MDL criterion~(\ref{eqn:MDL}).
\begin{enumerate}
\item Given a set of change points, apply the image segmentation method to all the images belonging to the first homogeneous time interval and obtain the MDL best-fitting image for this interval.  Repeat this for all remaining intervals.  Calculate the MDL criterion~(\ref{eqn:MDL}). 
\item Modify the set of change points by, for example, adding or removing one change point.  In terms of what modification should be made, we use the greedy strategy to select the one that achieves the largest reduction of the overall MDL value in~(\ref{eqn:MDL}).
\end{enumerate}
For Step~1 we begin with a large set of change points (i.e., \rev{an} over-fitted model). In Step~2 we remove one change point (i.e., merge two neighboring intervals) to maximize the reduction of the MDL value. The procedure stops and declares the optimization is done if no further MDL reduction can be achieved by removing change points. This is similar to backward elimination for statistical model selection problems. See Figure~\ref{fig:flowchart} for a flowchart of the whole procedure.

\subsubsection{Practical Considerations}
\label{sec:practical}

Here we list some practical issues that are crucial to the success of the above minimization algorithm. 

{\em Initial seed allocation for SRG:} 
The selection of the initial seeds in SRG plays an important role in the performance of the algorithm. To obtain a good initial oversegmentation, there must be at least one seed within each true region. Currently we use all the local maxima as well as a subset of the square lattice as the initial seeds. Based on simulation results (see Section~\ref{subsection:spatial_simulation}), when the number of initial seeds is inadequate, the SRG will underfit the images, which will in turn lead to an overfitting of the change points and will lead to an increased false positive rate. 
On the other hand, it could be time-consuming when the number of initial seeds is very large, especially for high resolution images. 
\rev{We developed an algorithm that allocates initial seeds automatically based on locating local maxima. However,} an optimal selection of initial seeds almost certainly requires expert intervention, as it depends on the type of data that we work with. To reduce the chance of obtaining a poor oversegmentation with SRG, in practice one could try using different sets of initial seeds and select the oversegmentation that gives the smallest MDL value.

{{\em Counts per bin:}} 
The photon counts in the image pixels cannot be too small, otherwise the algorithm could fail to produce meaningful output; see Section~\ref{subsection:spatial_simulation}. In some sense, small photon counts can be seen as low signal level, which means that the proposed method requires a minimum level of signal to operate with.  Therefore, care must be exercised when deciding the size for the bin.  As a rule of thumb, it should be enough to have around 100 counts for each pixel belonging to an astronomical object, while pixels from the background can have very low or even zero count.

{\em Initial change point selection:} 
Although the stepwise greedy algorithm is capable of saving a significant amount of computation time, it could still be time consuming if the initial set of change points is too large, as it might need many iterations to reach a local minimum.  It is recommended to select the initial change points based on prior knowledge, if available, in order to accelerate the algorithm.

{\em Computation Time:} 
In each iteration, the main time-consuming part is to apply the SRG. When the total number of pixels and the number of seeds are large, during the process, the number of candidates is large. Therefore, the comparison among all the candidates and the following updating manner lead to most of the computation burden. As an example, we found that it takes about 40 minutes for applying SRG and merging once on $64 \times 64$ images with about 200 seeds on a Linux machine with an octa-core 2.90 GHz Intel Xeon processor.

\subsection{Highlighting the Key Pixels}
\label{subsection:highlight}
After the change points are located, it is necessary to locate the pixels or regions that contribute to the estimation of the change points.  
The manner by which such {\em key pixels} are identified depends on the scientific context.  Below we present two methods that are applicable to the real-world examples we discuss in Section~\ref{section:real_data}

We focus here on images in a single passband, i.e., grey-valued images. For multi-band images, one can first transform a multi-band image into a single-band image by, for example, summing the pixel values in different bands, or by using only first principal component image of the multi-band image. Alternatively, one can also apply the method to each band individually, and merge the results from each band.

\subsubsection{Based on Pixel Differences}
\label{subsubsection:highlight_difference}

The first method is to highlight key pixels based on the distribution of pixel differences before and after change points. The rationale is that a pixel with different fitted values before and after a change point is a strong indicator that it is a key pixel.

Suppose the fitted values for pixel $i$ in time intervals $k$ and $(k+1)$ are $\hat{\lambda}_i^{(k)}$ and $\hat{\lambda}_i^{(k+1)}$, respectively. Given the Poisson nature of the data, we first apply a square-root transformation to normalize the fitted values. Define the difference $d_i$ for pixel $i$ as 
\begin{equation}
    d_i = \sqrt{\hat{\lambda}_i^{(k+1)}} - \sqrt{\hat{\lambda}_i^{(k)}} \,.
\end{equation}
A pixel is labelled as a key pixel if its $d_i$ is far away from the mean of all the differences. To be specific, pixel $i$ is labelled as a key pixel if
\begin{equation}
    \left| \frac{d_i - \hat{\mu}}{\hat{\sigma}} \right| > \Phi^{-1}\left(1-\frac{1}{2} p\right),
\end{equation}
where $\hat{\mu} = \frac{1}{\nx} \sum_{1}^{\nx} d_i$ and $\hat{\sigma} = \frac{1}{\Phi^{-1}(3/4)} \text{MAD}$.  Here $\text{MAD} = \text{median}(| d_i - \tilde{d} |)$ is the median absolute deviation with $\tilde{d} = \text{median}(d_i)$, and is used to obtain a robust estimate (i.e., a measure that minimizes the effect of outliers) of the standard deviation of the $d_i$'s \citep[see e.g.,][]{doi:10.1080/01621459.1993.10476408}. $\Phi^{-1}(\cdot)$ is the quantile of the standard normal distribution, and $p$ is the pre-specified significance level\footnote{$\Phi^{-1}$ is related to the standard Normal error function, with exemplar values $\Phi^{-1}(\{0.75, 0.841, 0.977, 0.9986, 1-\frac{10^{-5}}{2}, 1-\frac{10^{-10}}{2}, 1-\frac{10^{-15}}{2}\})=\{0.6745, 1, 2, 3, 4.417, 6.467, 8.014\}$.  We typically choose $p=1-\frac{10^{-15}}{2}$ as our threshold.}. Notice that by checking the sign of $d_i - \hat{\mu}$, we can deduce if pixel $i$ has increased or decreased after this change point.

\subsubsection{Based on Region Differences}
\label{subsubsection:highlight_testing}

Another method to locate key pixels is to compare pairs of regions. For any region in the time interval $k$, there must exist at least one region in time interval $(k+1)$ such that these two regions have overlapping pixels. We then test if the difference between the means of the pixels from these two regions is significant or not.  

As before, we apply the square-root transformation to the pixels within each of the regions. Then we calculate the sample means $\hat{\mu}_1$ and $\hat{\mu}_2$ and sample variances $\hat{\sigma}_1^2$ and $\hat{\sigma}_2^2$ of these two groups of square-rooted values.
Then we can for example test \rev{whether} the difference between $\hat{\mu}_1$ and $\hat{\mu}_2$ is large enough with 
\begin{eqnarray}
    \left| \frac{\hat{\mu}_2 - \hat{\mu}_1}{\sqrt{\hat{\sigma}_1^2 + \hat{\sigma}_2^2}} \right| > \Phi^{-1}\left(1-\frac{1}{2} p\right).
\end{eqnarray}

See Section~\ref{section:real_data} for the applications of these two methods on some real data sets.

Lastly we note that the selection of $p$, the significance level, deserves a more careful consideration.  As in reality one may need to do comparisons for many change points and energy bands, this becomes a multiple-testing problem where the number of tests is large.  Therefore, one should adjust the value of $p$ in order to control false positives.

\section{Simulations}

Two groups of simulations were conducted to evaluate the empirical performance of the proposed method.  A specially designed $\lambda_{i,t,w}$ was used for each of the experiments. For each experiment, we tested 13 signal levels, defined as the average number of photon counts per pixel.  For each signal level, 100 datasets were generated according to~(\ref{equation:poisson}), with $\Delta T = 1$.  The number of spectral bands $N_W=3$ so $w=1, 2, 3$ and the number of time points $N_T=60$.

\subsection{Group 1: Single Pixel}

The first group of experiments were designed to evaluate the ability of the proposed method for detecting change points, under the condition that there are no spatial variations. To be more specific, the size of the images is $1 \times 1$; i.e., only one pixel.  In other words, $\nx=1$ and the $i$ in $\lambda_{i,t,w}$ is a dummy index. Three $\lambda_{i,t,w}$'s of increasing complexity were used:
\begin{enumerate}
    \item $\lambda_{i,t,w}$ is constant; i.e., no change point {(as depicted in Figure~\ref{fig:simulation_model} \rev{(a)})}.
    \item $\lambda_{i,t,w}$ shows intensity changes but all three bands are identical at any given $t$ {(see Figure~\ref{fig:simulation_model} \rev{(b)})}.
    \item $\lambda_{i,t,w}$ shows spectral changes {(see Figure~\ref{fig:simulation_model} \rev{(c)})}.
\end{enumerate}
The first $\lambda_{i,t,w}$ was used to study the level of false positives, while the remaining two were used to study false negatives. 

\begin{figure}[ht!]
\plotone{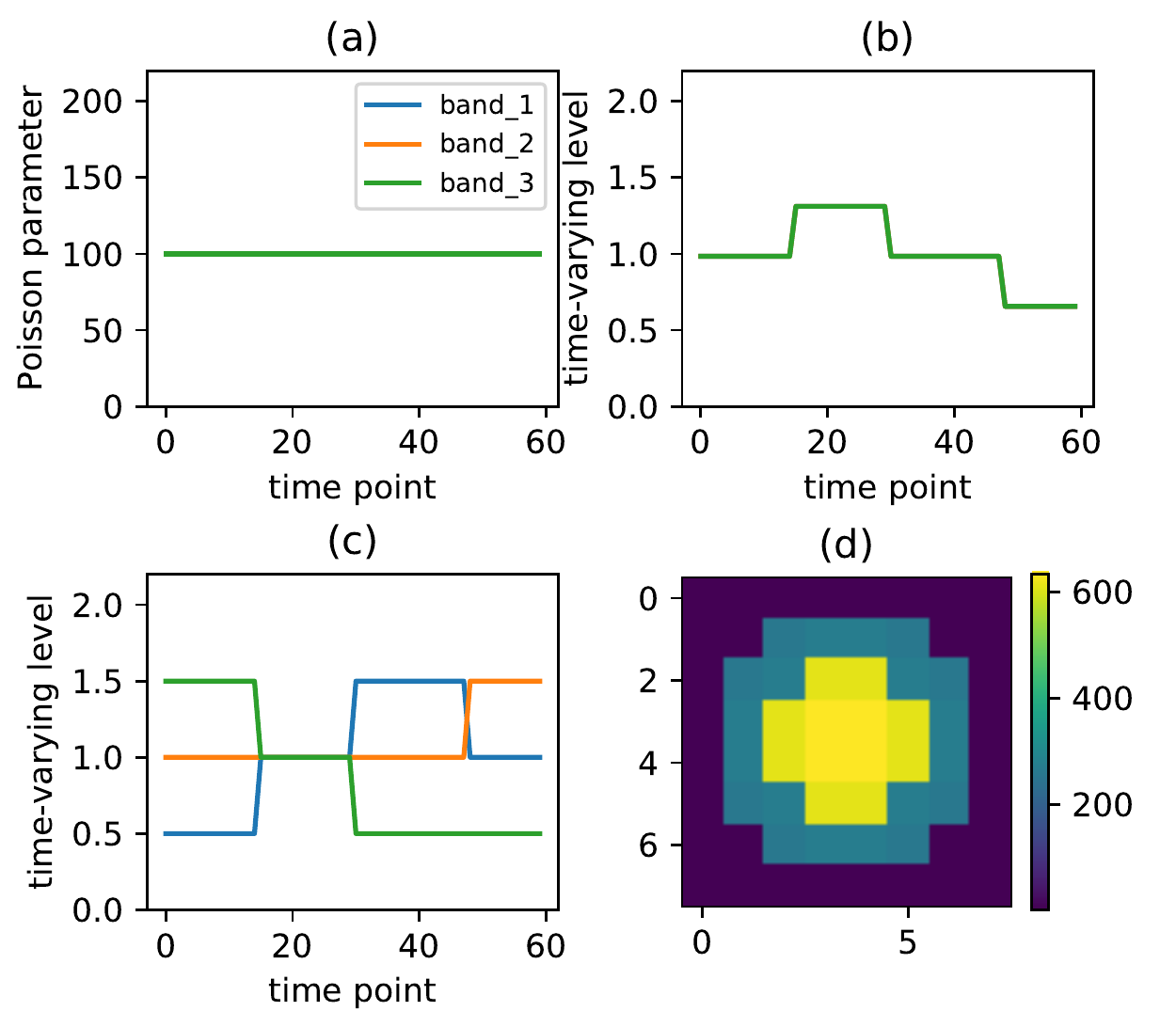}
\caption{
The Poisson rate functions $\lambda_{i,t,w}$ used in the simulation experiments. {(a):} $\lambda_{i,t,w}$ used in the single pixel experiment without change point (Section~\ref{sss:nochange}).  The x-axis denotes the time points, while the y-axis shows the values of $\lambda_{i,t,w}$ for different band $w$. 
{(b):} the $\lambda_{i,t,w}$ relative to the no-change case (top left), used in the single pixel simulation with changing intensity (Section~\ref{sss:varintens}). 
{(c):} the $\lambda_{i,t,w}$ for different passbands (marked in blue, orange and green) relative to the no-change case (top left), used in the single pixel simulation with changing intensity and spectra (Section~\ref{sss:varintspec}).
{(d):} the spatial structure used for the second group of experiments.  Size of the image is $m=n=8$.
}
\label{fig:simulation_model}
\end{figure}

\subsubsection{No Change Point}\label{sss:nochange}
As there is no change point in $\lambda_{i,t,w}$, this experiment is ideal for studying the relationship between the false positive rate and the signal level; recall the latter is defined as the average number of photon counts per pixel.

The results of this experiment (together with the next two experiments) are summarized {as the blue curves} in Figure~\ref{fig:simulation_single_pixel_result}. 
{The figure captures how well the simulation recovers the location of the change points (top left), the number of change points (top right), the excess number of change points (false positives; bottom left), and the deficit in change points (false negatives; bottom right).}
The top left plot reports the fraction of simulated datasets for which the set of the fitted change points $\hat{\bm{\tau}}$ is identical to the set of true change points $\bm{\tau}$. The top right plot presents the fraction of simulated datasets for which the fitted number of change points $\hat{K}$ equals to the true number of change points $K$.  The bottom left plot shows the average false positive rate, which is defined as the average number of falsely detected change points per possible location.  The bottom right plot presents the fraction of simulated datasets for which $\hat{\bm{\tau}}$ contains $\bm{\tau}$, i.e., $\bm{\tau} \subseteq \hat{\bm{\tau}}$. One can see that the false discovery rate seems to be quite stable across different signal levels.  Notice that the last curve is always 1 because $\bm{\tau}$ is empty for this experiment.

\begin{figure}[ht!]
\plotone{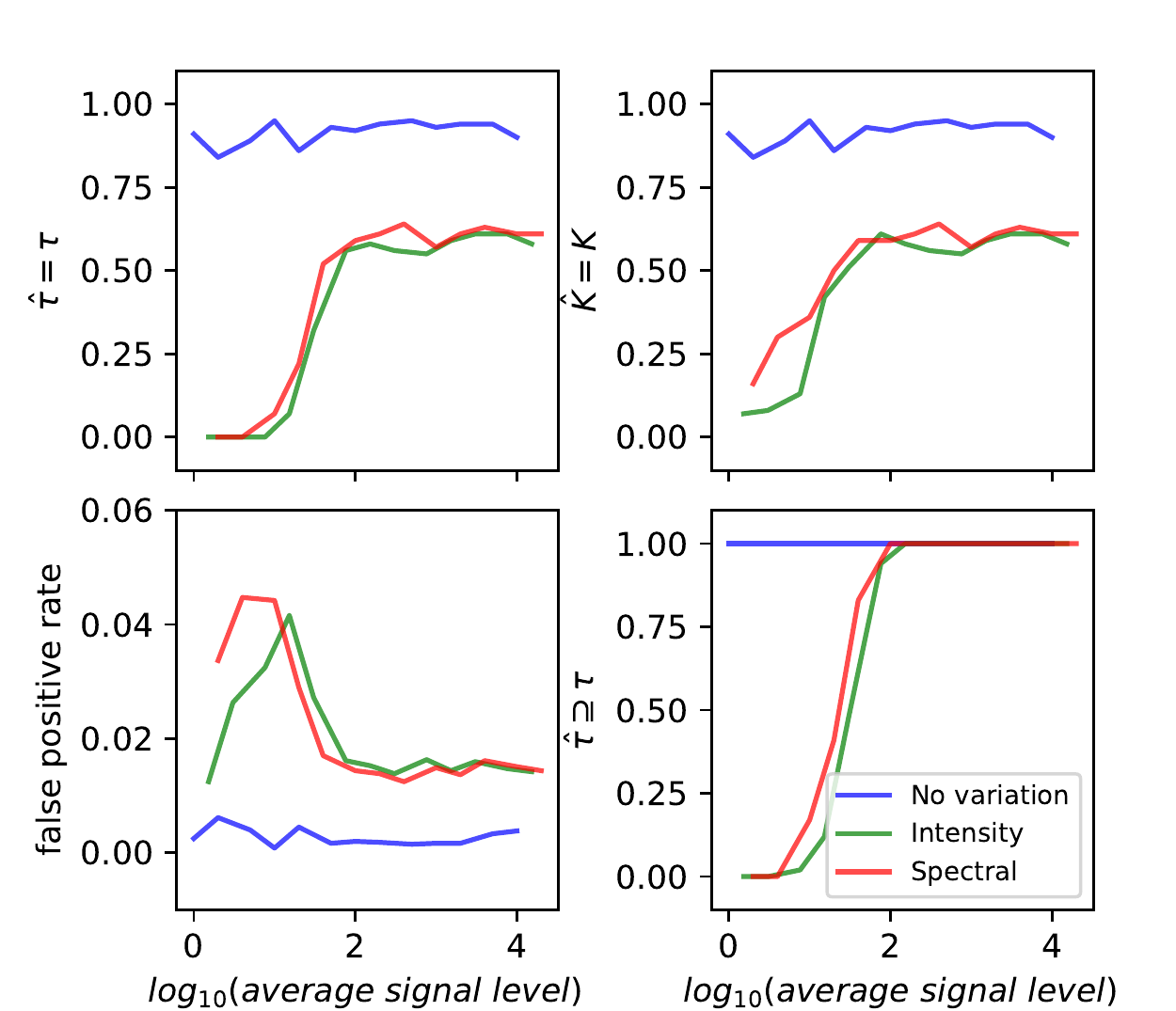}
\caption{Simulation results for Group 1 experiments (single pixel). For all four plots, the x-axes denote the logarithm of the average counts of photons over all time points and all bands. {\sl Top left:} fraction of the fitted change points $\hat{\bm{\tau}}$ that are identical to the true change points $\bm{\tau}$. {\sl Top right:} fraction of the fitted number of change point $\hat{K}$ that equals to the true number of change point $K$. {\sl Bottom left:} false positive rate. {\sl Bottom right:} fraction of fitted change points contains true change points; i.e., $\bm{\tau} \subseteq \hat{\bm{\tau}}$.  Note that the legend in the bottom right plot holds for all four plots.}
\label{fig:simulation_single_pixel_result}
\end{figure}

\subsubsection{Varying Intensity}\label{sss:varintens}

In this experiment we introduced variation in $\lambda_{i,t,w}$ by multiplying \rev{(a) and (b)} in Figure~\ref{fig:simulation_model} together. The results are reported as the green curve in Figure~\ref{fig:simulation_single_pixel_result}. One can see that when the signal level is small ($\log_{10}(\text{average signal level})<1.0$), increasing the signal level leads to more false positives: this range of signals levels are too small to provide enough information for the detection of the true change points. 

When $\log_{10}(\text{average signal level})$ is between $1.0$ and $2.0$, the proposed method starts to be able to detect the true change points as the signal level increases. Also, the false positive rate begins to drop.  One possible explanation is that for any two consecutive homogeneous time intervals, the location of the change point might not be clear when the signal level is relatively small.

As the signal level continues to increase, the proposed method becomes more successful in detecting the true change point locations. For example, when the average number of counts in each bin is greater than 100 (i.e., $\log_{10}(\text{average signal level})>2.0$), the signal is strong enough so that all the true change points can be detected successfully.  We note that there are always some false positives due to the Poisson randomness, and it seems that the false positive rate stabilizes as the signal level increases.

\subsubsection{Varying Spectrum}\label{sss:varintspec}

Here we allow different bands $w$ \rev{to} change differently at the change points. The rate $\lambda_{i,t,w}$ was obtained by multiplying \rev{(a) and (c)} in Figure~\ref{fig:simulation_model} \rev{together}. The results, which are about the same as the previous experiment (varying intensity), are reported as the red curve in Figure~\ref{fig:simulation_single_pixel_result}. When the signal level is small ($\log_{10}(\text{average signal level}) < 1.0$), the proposed method fails to detect the true change points.  When $\log_{10}(\text{average signal level})$ is between $1.0$ and $2.0$, as the signal level increases, the false positive rate begins to decrease while the true positive rate increases. When $\log_{10}(\text{average signal level}) > 2.0$, all the true change points can be detected successfully, while the false positive rate stays at the same level as signal level increases. 

\subsection{Group 2: Spatial Structure}
\label{subsection:spatial_simulation}
Instead of having a constant spatial signal (i.e., single pixel), in this second group of experiments a spatial varying structure is introduced to study the empirical performance of the proposed method. As before, three Poisson rate functions $\lambda_{i,t,w}$ are considered.  The size of the image is set to $N_I=8\times 8$.  To illustrate the importance of initial seed placement, we tested two allocation strategies: (i) we deliberately placed \rev{an} inadequate number of initial seeds and (ii) we used every pixel as an initial seed.

\subsubsection{No Change Point}

There was no change point in this experiment and {the spatial variation of} $\lambda_{i,t,w}$ is given in the bottom right plot of Figure~\ref{fig:simulation_model}. The results are reported \rev{as the blue curves} in Figure~\ref{fig:simulation_spatial_result}. One can see that if the number of initial seeds is inadequate, the false positive rate increases as the signal level increases above $\log_{10}(\text{average signal level}) > 2.5$. However, this does not happen when there are a large number of initial seeds; \rev{see the blue dotted curves in Figure~\ref{fig:simulation_spatial_result}.} In fact, for this and the following two experiments, our method did not detect any false positive change points. This suggests that when the images are under-segmented, the method tends to place more false change points to compensate for data variability not explainable by image segmentation.

\begin{figure}[ht!]
\plotone{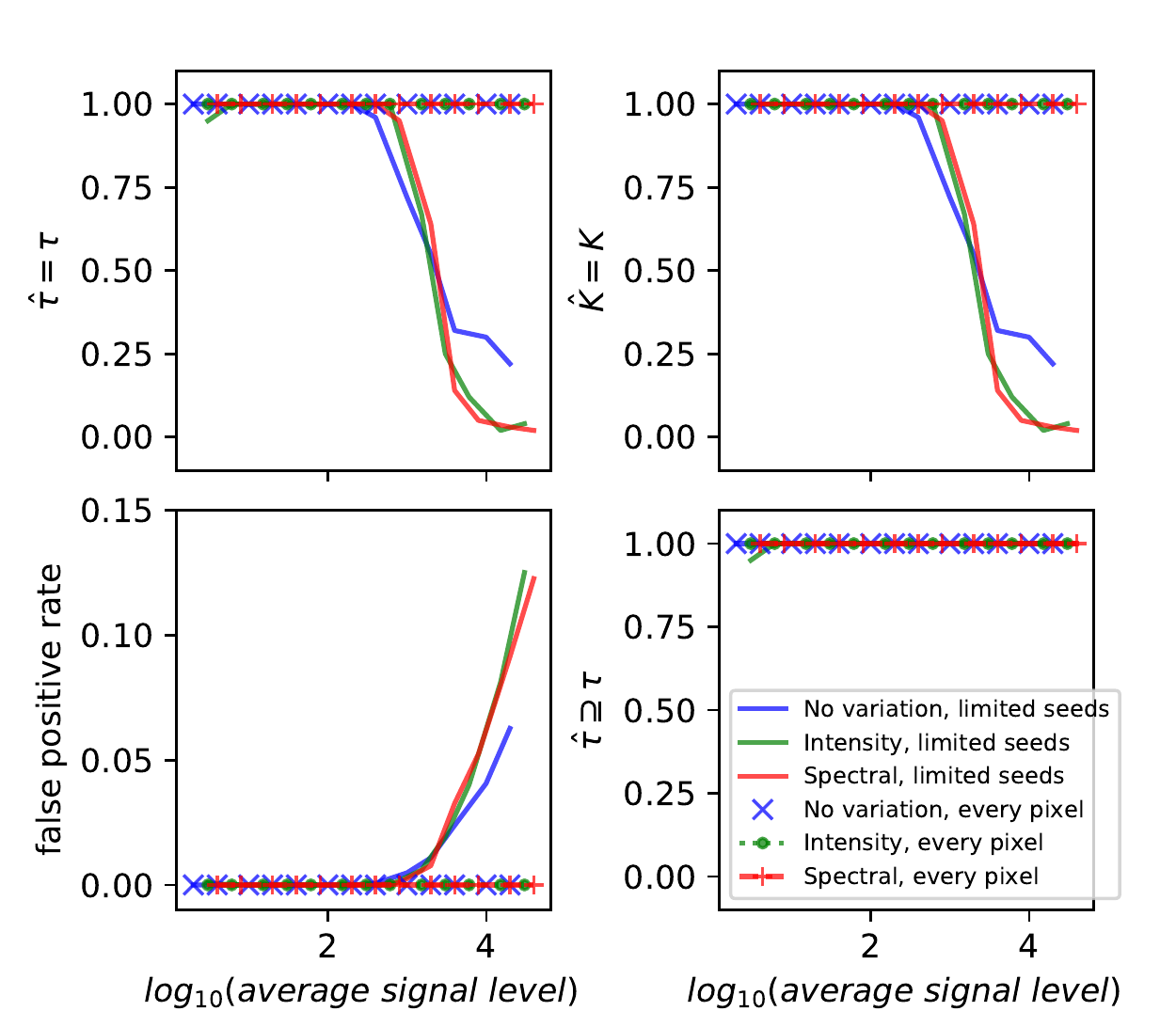}
\caption{Simulation results for Group 2 experiments (with spatial structure). For all four plots, the x-axes denote the logarithm of average count of photons over all time points and all bands. Solid curves denote the results when the number of initial seeds was inadequate, while the dotted curves show the results when every pixel was assigned as a initial seed. {\sl Top left:} fraction of the fitted change points $\hat{\bm{\tau}}$ that are identical to the true change points $\bm{\tau}$. {\sl Top right:} fraction of the fitted number of change points $\hat{K}$ that equals to the true number of change points $K$. {\sl Bottom left:} false positive rates. {\sl Bottom right:} fraction of fitted change points contains the true change points; i.e., $\bm{\tau} \subseteq \hat{\bm{\tau}}$. The legend in the bottom right plot holds for all these four plots.}
\label{fig:simulation_spatial_result}
\end{figure}

\subsubsection{Varying Intensity}

In this experiment $\lambda_{i,t,w}$ was obtained by multiplying \rev{(b) and (d)} of Figure~\ref{fig:simulation_model} together, so there are three change points over time.  The results are similar to the no change point case, and rev{summarized as the green curves} in Figure~\ref{fig:simulation_spatial_result}.

\subsubsection{Varying Spectrum}

In this last experiment the energy bands were allowed to be different, and $\lambda_{i,t,w}$ was obtained by multiplying \rev{(c) and (d)} of Figure~\ref{fig:simulation_model} together. The results, reported \rev{as the red curves in} Figure~\ref{fig:simulation_spatial_result}, are similar to the previous two experiments.

\subsection{Empirical Conclusions}
The following empirical conclusions can be drawn from the above experimental results.
\begin{itemize}
\item The method works well in all cases when the signal level is 
sufficiently large. As a rule of thumb, \rev{for binning of the original data,} it would be ideal to have 100 counts or more for each bin covering an astronomical source.
\item It is important to place enough initial seeds when applying SRG; otherwise the false positive rate will increase with the signal level. 
{See the second paragraph of Section~\ref{sec:practical} for some practical guidelines for initial seed selection.}
\end{itemize}

\section{Applications to Real Data}
\label{section:real_data}

To illustrate the usage in the astrophysics field, we apply the proposed method on two real datasets, which are more complicated than those in the previous section. Specifically, we select these datasets with some obvious time-evolving variations to demonstrate the performance of our method.

\subsection{XMM-Newton Observations of Proxima Centauri}

\begin{figure}[ht!]
\plotone{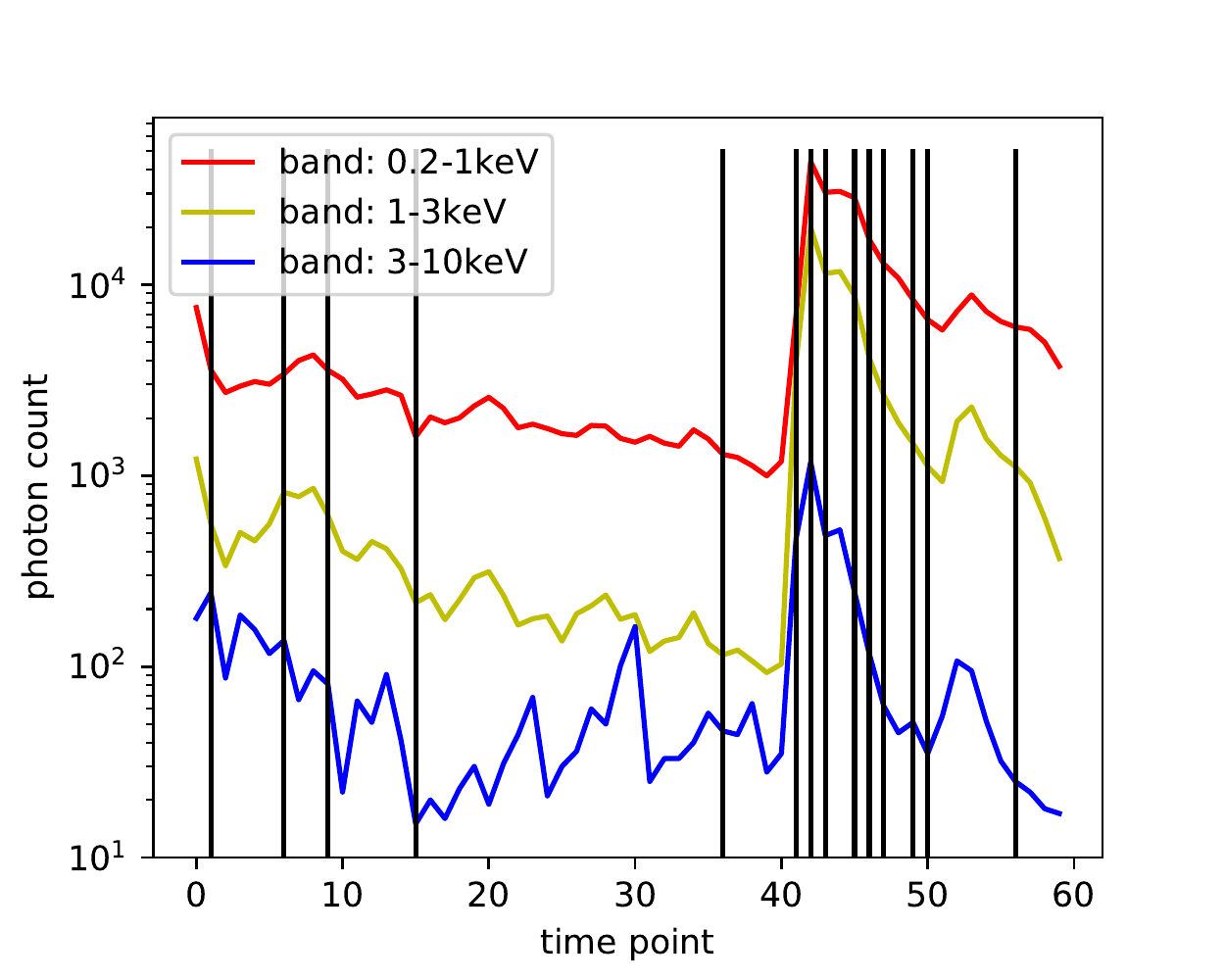}
\caption{Light curves of Proxima~Cen in different bands. Each curve denotes the number of photons within the corresponding band at a given time point index. Vertical black lines denote the locations of the detected change points.}
\label{fig:Proxima_Cen_light_curve}
\end{figure}

\object{Proxima Centauri} is the nearest star to the Sun and as such is well suited for studies of coronal activity. Like our Sun, Proxima Centauri operates an internal dynamo, which generates a stellar magnetic field. In the standard model for stellar dynamos, the magnetic field lines wind up through differential rotation. When some of the magnetic field lines reconnect, the energy is released in the stellar flare. Such flares typically show a sudden rise in X-ray emission and a more gradual decay over several hours. In flares, flux and temperature are correlated such that a higher X-ray flux corresponds to a higher temperature and thus a higher energy of the average detected photons \citep[see][for a review of X-ray emission in stellar coronae and further references]{Guedel_review}.

Despite its proximity, Proxima Centauri and its corona are unresolved in X-ray observations; it is just the point-spread function (PSF) of the telescope that distributes the incoming flux over many pixels on the detector.

\subsubsection{Data}
We use a dataset from XMM-Newton (Obs.ID 0049350101), where Proxima Centauri was observed for 67~ks on 2001-08-12. Because of the high flux, the MOS cameras on XMM-Newton are highly piled-up and we restrict our analysis to the data from the PN camera. We obtained the data from \rev{the XMM-Newton science archive hosted by the European Space Agency (ESA)\footnote{\url{https://www.cosmos.esa.int/web/xmm-newton/xsa}}. The data we received was processed by ODS version 12.0.0.}
Our analysis is based on the filtered PN event data from the automated reduction pipeline (PPS). \citet{Guedel2002} presented a detailed analysis and interpretation of this dataset.

In our analysis, we only used a subset of photons with spatial coordinates within $[25500, 27500] \times [26500, 28500]$, and it was binned as images of size $64 \times 64$.  We used the temporal bins of width $1100.4$ seconds to generate $60$ images.  We binned the data into three energy bands, $(200,1000]$, $(1000, 3000]$ and $(3000,10000]$ in eV.

\subsubsection{Results}

\begin{figure}[ht!]
\plotone{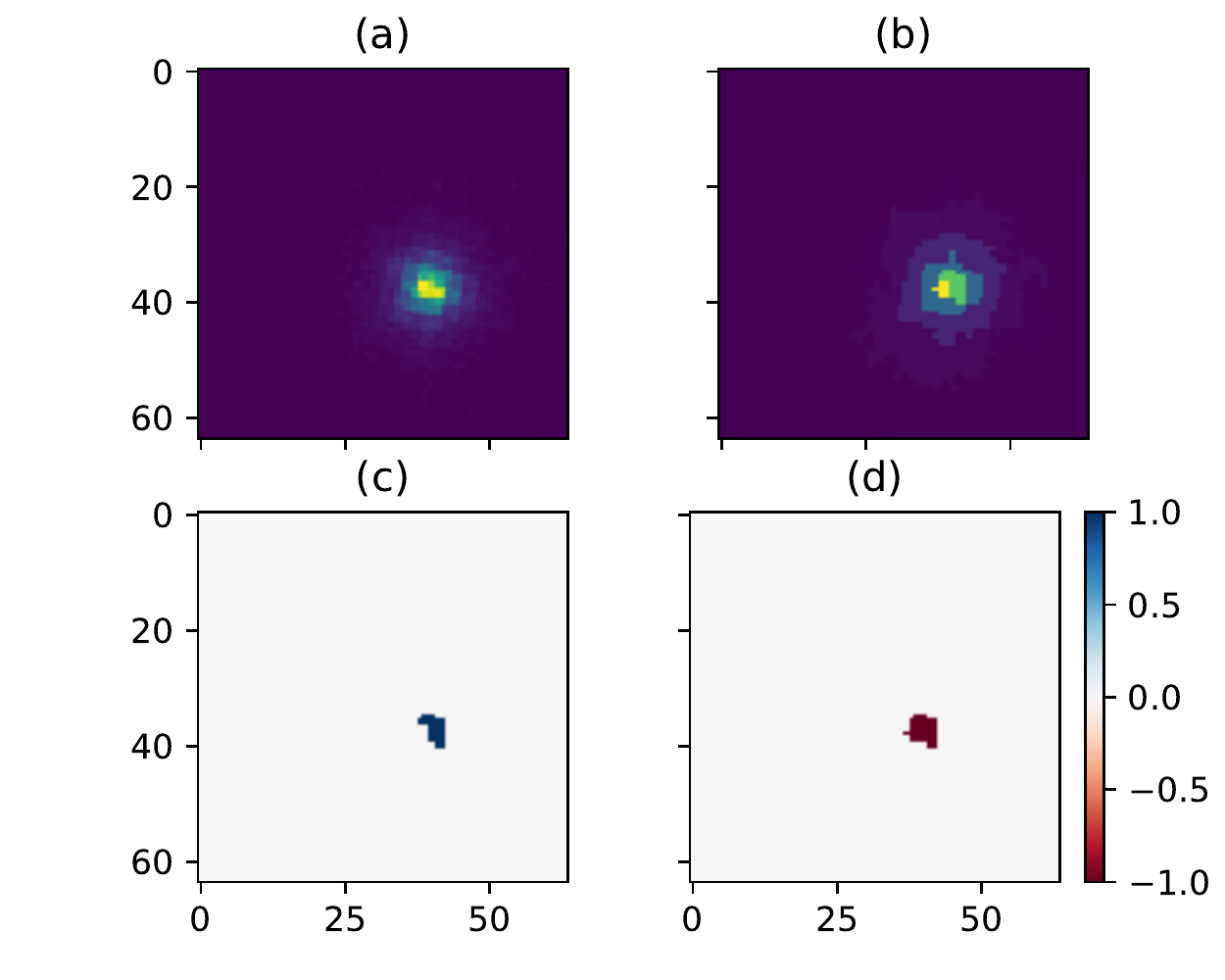}
\caption{Results for Proxima Centauri. (a): the data image at time point 42 for the first band (200, 1000] in eV. (b): the corresponding fitted value $\lambda_{i,t,w}$. (c): 
\rev{regions that show an increase (blue) and decrease (red) in intensity prior to this time point.} 
\rev{Compared with the previous time interval, there was a significant increase in the source at this time point.} (d):
\rev{as in panel {\sl c}, but for the epoch after this time point.}
\rev{After this time point, the brightness in the source decreased.} Notices that these two bottom plots share the colorbar, \rev{where the} 
value $1$ denotes increasing and $-1$ denotes decreasing
\rev{intensities}.}
\label{fig:Proxima_Cen_result}
\end{figure}

Figure~\ref{fig:Proxima_Cen_light_curve} presents the light curves for different bands as well as the locations of the detected change points. As there is only a single source of photons with negligible background signal level, the detected change points coincide with the changes of the light curves. Many change points are detected for the abrupt increase and then decrease in brightness for all the bands at the time points between 42 and 50. The time interval between 15 and 36 is detected as a homogeneous time interval, and the variation in light curves within this interval is viewed as common Poisson variations. A few change points are detected for the time interval before time point 15 and the interval after 50. A piecewise constant model is used to fit these gradual changes in intensities.

The fitted images can be found in Figure~\ref{fig:Proxima_Cen_result}. The source of most photons, a point source, is modeled by different piecewise constant models at all these time intervals. The center of the point source and the wings of the PSF region nearby were fitted by models with shapes like concentric circles.

To test our method for finding the regions of significant change, we also apply it here because we know what the answer should be. For this dataset, as the observations for different bands change simultaneously, we combine all the three bands to highlight the key pixels. That is, we highlight the regions for each band and take the intersection of these regions. Examples of the results based on the method in Section~\ref{subsubsection:highlight_testing} can be found in Figure~\ref{fig:Proxima_Cen_result}. The method indeed picks out the point source. With significance level \rev{$p = 10^{-2}$}, the abrupt increase in brightness of the source at time point 41 and 42, as well as the sudden decrease at time point 43 can be detected successfully. By modifying the significance level, different sensitivity can be achieved.

\subsection{Isolated Evolving Solar Coronal Loop}

\begin{figure*}[h!]
    \centering
    \includegraphics[width=5in]{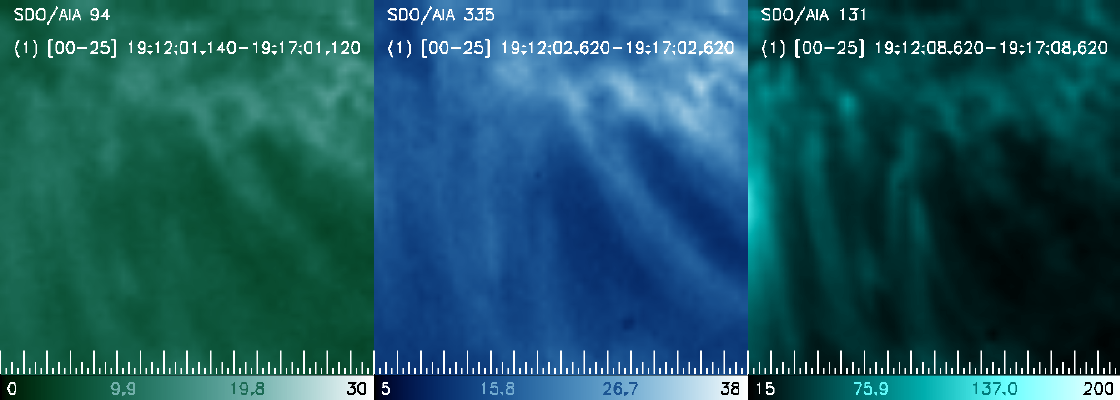}
    \includegraphics[width=5in]{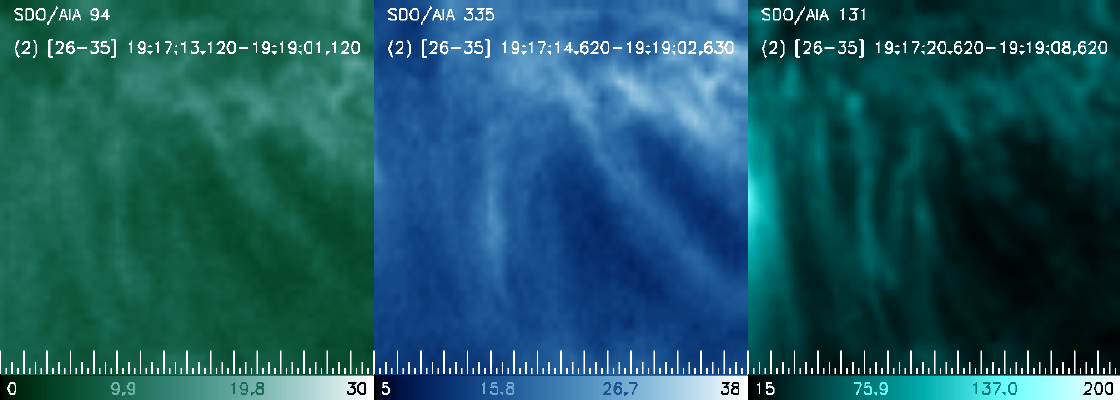}
    \includegraphics[width=5in]{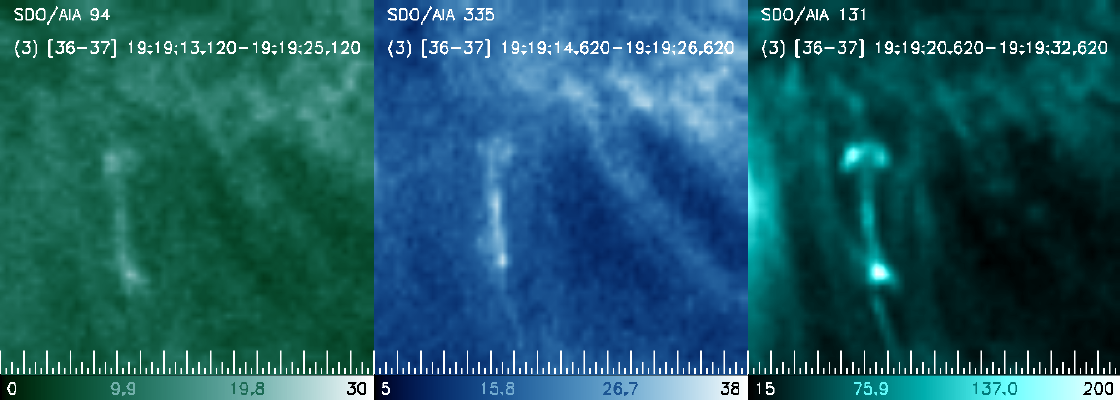}
    \includegraphics[width=5in]{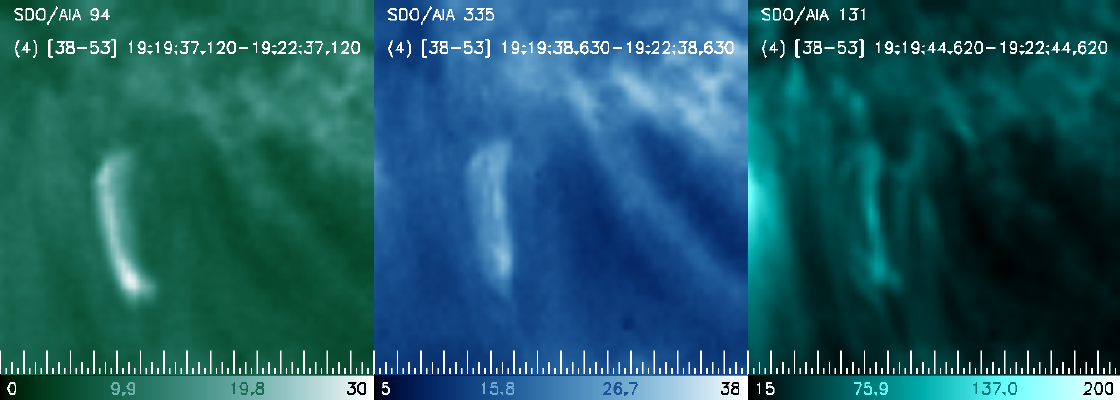}
    \caption{An isolated loop structure shown lighting up in 3 SDO/AIA passbands.  Each row corresponds to the intensities in AIA filter images, averaged over the time duration found by our method, going from interval $1$ (top) to interval $4$ (bottom).  The columns, going from left to right, show the 94, 335, and 131~\AA\ filter band images.  The filter name, time duration, and the image sequence indices are marked at the top of \rev{the} image and the intensity scale is marked at the bottom.  The grid at the bottom of each image denotes the pixelation, with each image having a size of $64{\times}64$ pixels.   Notice that the isolated loop becomes bright enough for detection in the $3^{\rm rd}$ interval.}
    \label{fig:AIA3obs}
\end{figure*}

\begin{figure*}[ht!]
    \centering
    \includegraphics[width=5in]{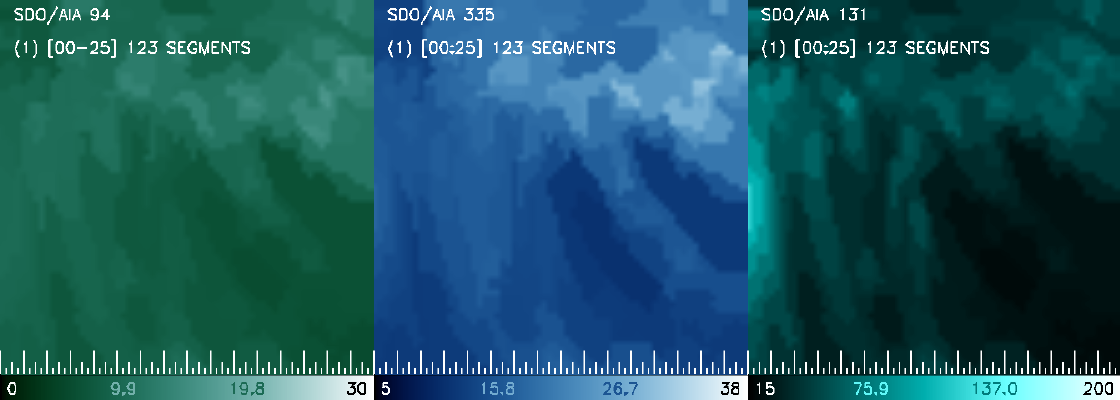}
    \includegraphics[width=5in]{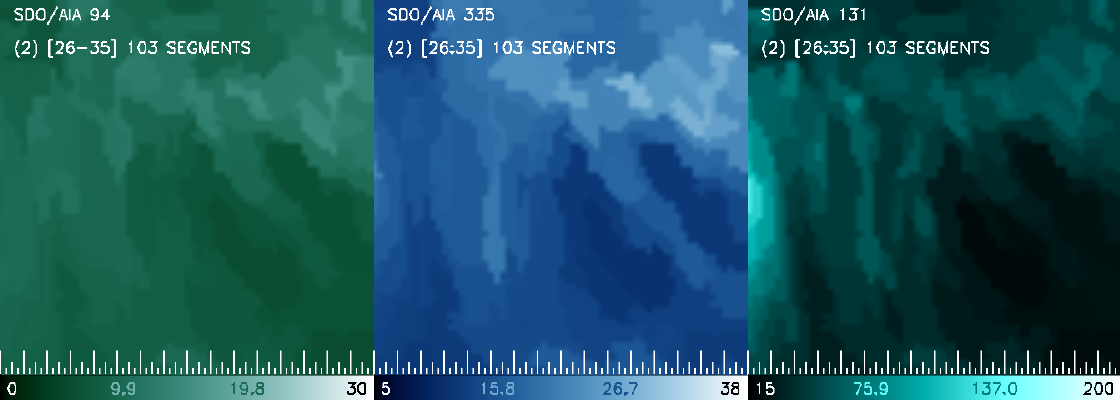}
    \includegraphics[width=5in]{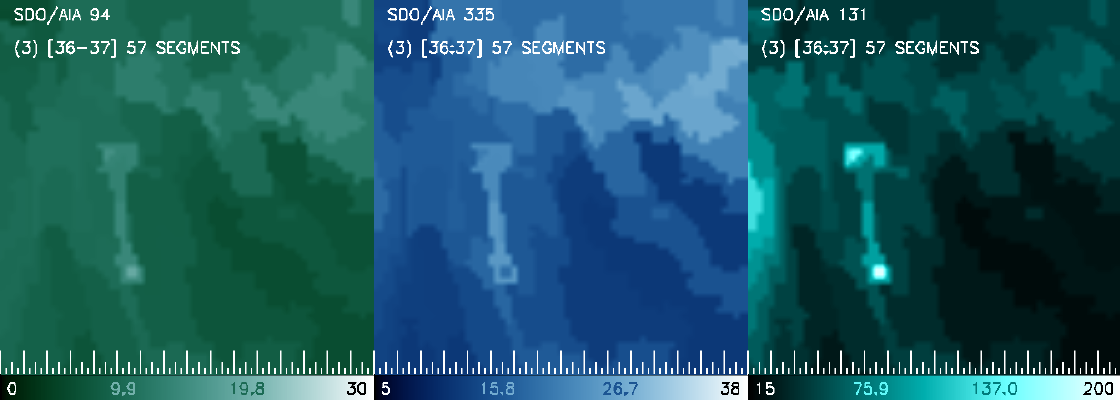}
    \includegraphics[width=5in]{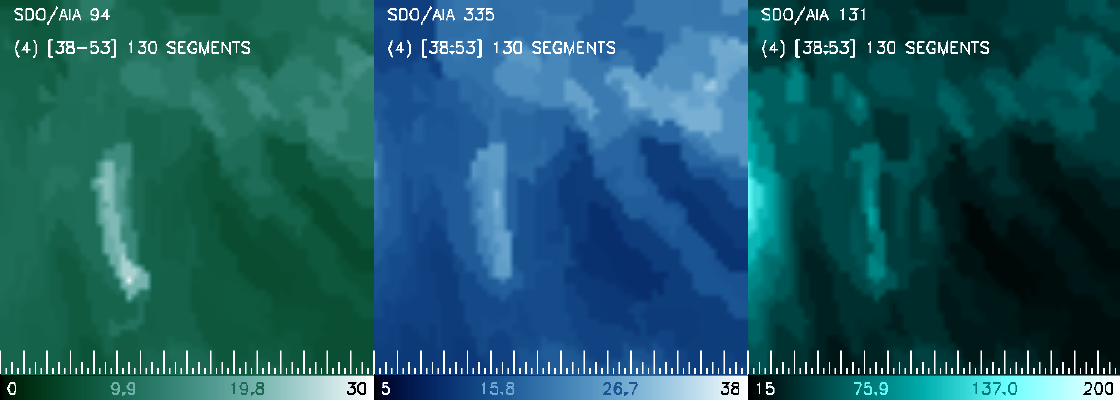}
    \caption{Intensities $\lambda_{i,t,w}$ as fit to the data from Figure~\ref{fig:AIA3obs} in spatial segments.  The images are arranged in the same manner, and demonstrate that the loop structure is locatable and identifiable.  The number of region segments found are also marked.}
    \label{fig:Solar_fitted_3band}
\end{figure*}

Images of the solar corona constitute a legitimate Big Data problem.  Several observatories have been collecting images in extreme ultra-violet (EUV) filters and in X-ray passbands for several decades, and analyzing them to pick out interesting changes using automated routines have been largely unsuccessful.  Catalogs like the HEK \citep[Heliophysics Events Knowledgebase][]{2012SoPh..275...67H,2012SoPh..275...79M} can detect and mark features of particular varieties, though these compilations remain beset by incompleteness \citep[see, e.g.,][]{2018ApJS..236...15A,2019JSWSC...9A..38H,2017IAUS..325..201B}.  In this context, our method provides a way to model solar features without limiting it to a particular feature set to identify and locate regions in images where something interesting has transpired.  As a proof of concept, we apply the method to a simple case of an isolated coronal loop filling with plasma, as observed with the Solar Dynamics Observatory's Atmospheric Imaging Assembly (SDO/AIA) filters \citep[][]{2012SoPh..275....3P,2012SoPh..275...17L}.  Considerable enhancements must still be made \rev{in order to lower the} computational cost before the method can be applied to full size images at faster than observed cadence; however, we demonstrate here that a well-defined region of interest can be selected without manual intervention for a dataset that consists of images in several filters.

\subsubsection{Data}

In particular, here we consider AIA observations carried out on 2014-Dec-11 between 19:12 UT and 19:23 UT, and focus on a $64{\times}64$ pixel region located $(+1'',-271'')$ from disk center, in which a small, isolated, well-defined loop appeared at approximately 19:19 UT.  This region was selected solely as a test case to demonstrate our method; the appearance of the loop is clear and unambiguous, with no other event occurring nearby to confuse the issue; see Figure~\ref{fig:AIA3obs}.  We apply our method to these data\rev{, downloaded using the SDO AIA Cutout Service,\footnote{\rev{\url{https://www.lmsal.com/get_aia_data/}}}} and demonstrate that the loop (and it alone) is detected and identified; see Figures~\ref{fig:Solar_fitted_3band}, \ref{fig:Solar_highighted_method3_4.1} and~\ref{fig:solar_keypix_lightcurves}.
AIA data are available in 6 filter bands, centered at 211, 94, 335, 193, 131, 171~\AA.  Here, we have limited our analysis to 3 bands: 94, 335, and 131~\AA{} in which the isolated loop is easily discernible to the eye (a full analysis including all the filters does not change the results).  Each filter consists of a sequence of 54 images, and while they are not obtained simultaneously, the difference in time between the bands is ignorable on the timescale over which the loop evolves.

\subsubsection{Results}

\begin{figure}[htb!]
    \centering
    \includegraphics[width=3.3in]{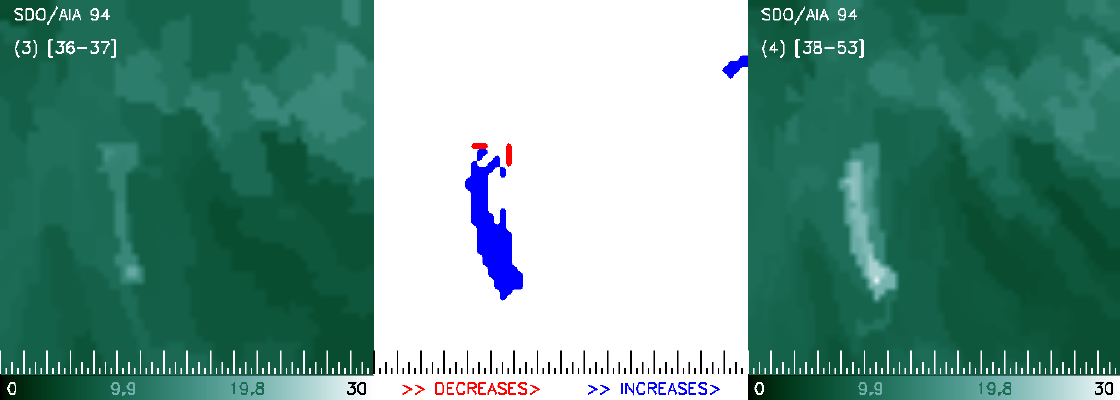}
    \includegraphics[width=3.3in]{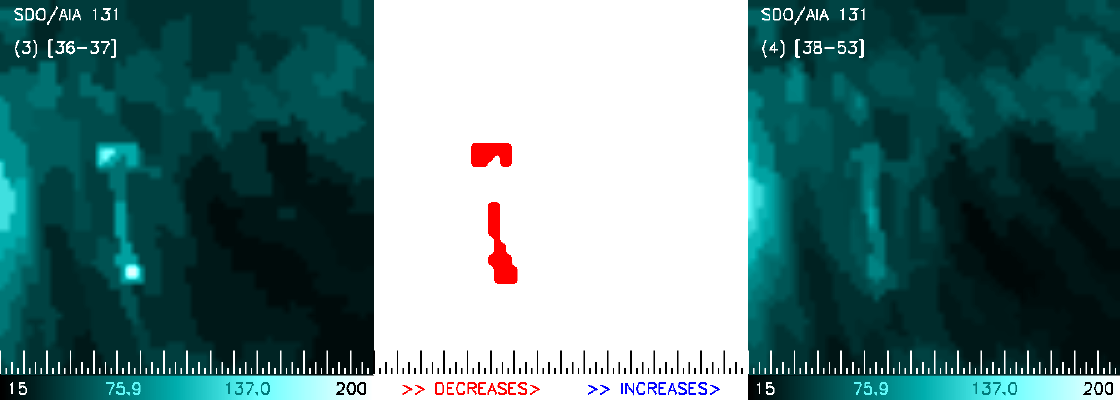}
    \caption{Demonstrating the isolation of key pixels of interest.  Each set of three shows the fitted intensity in one passband in the $3^{\rm rd}$ interval ({\sl left}), followed by a bitmap of pixels ({\sl middle}) showing where intensity increases (blue) and decreases (red), followed by the fitted intensity image in the same filter in the $4^{\rm th}$ time interval ({\sl right}).  The upper row shows the transition in the 94\,\AA\ filter, and the lower row shows the transition in the 131\,\AA\ filter.  Notice that the loop continues to brighten at 94\,\AA, even as it starts to fade at 131\,\AA.
    }
    \label{fig:Solar_highighted_method3_4.1}
\end{figure}

\begin{figure}[ht!]
    \centering
    \includegraphics[width=3.3in]{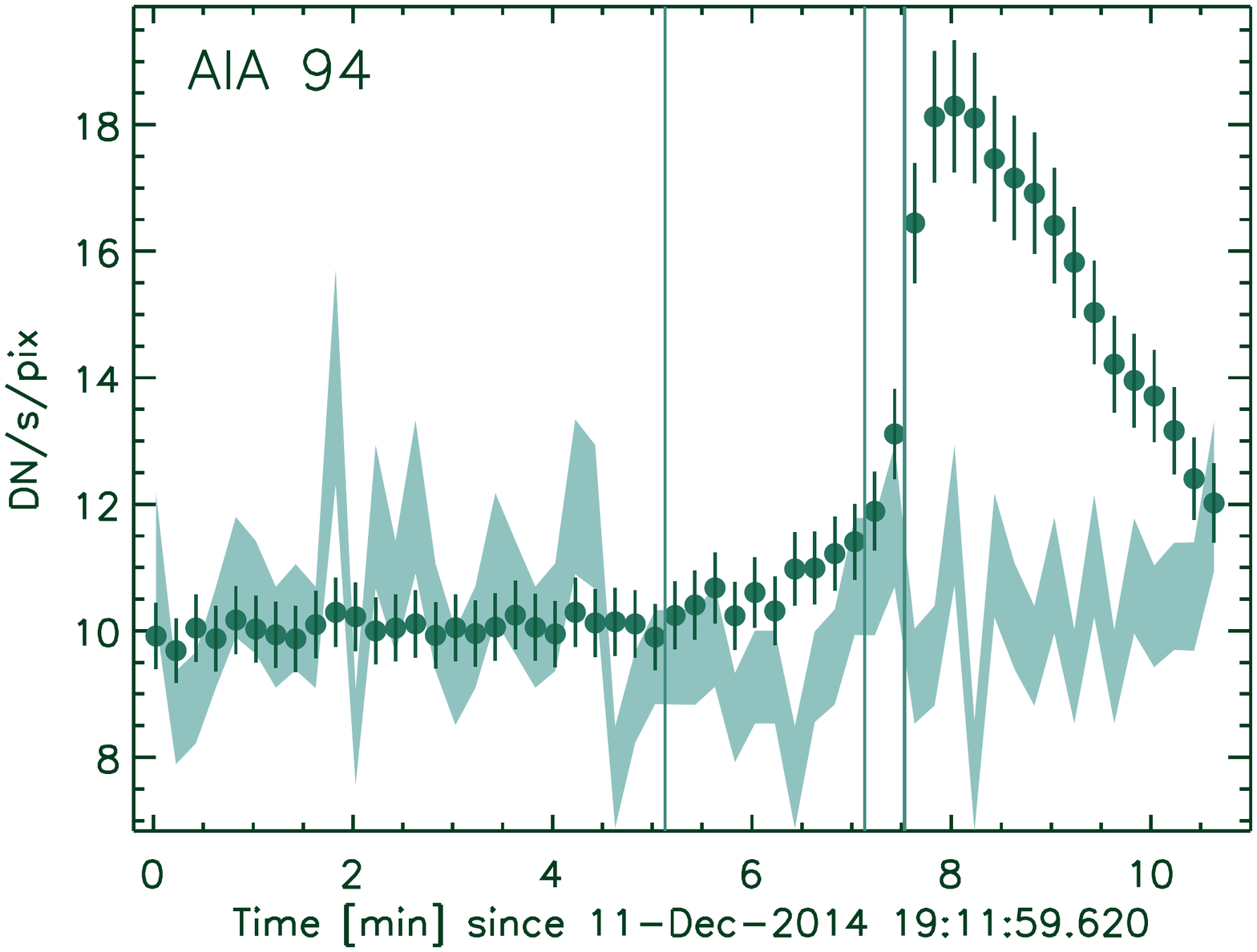}
    \includegraphics[width=3.3in]{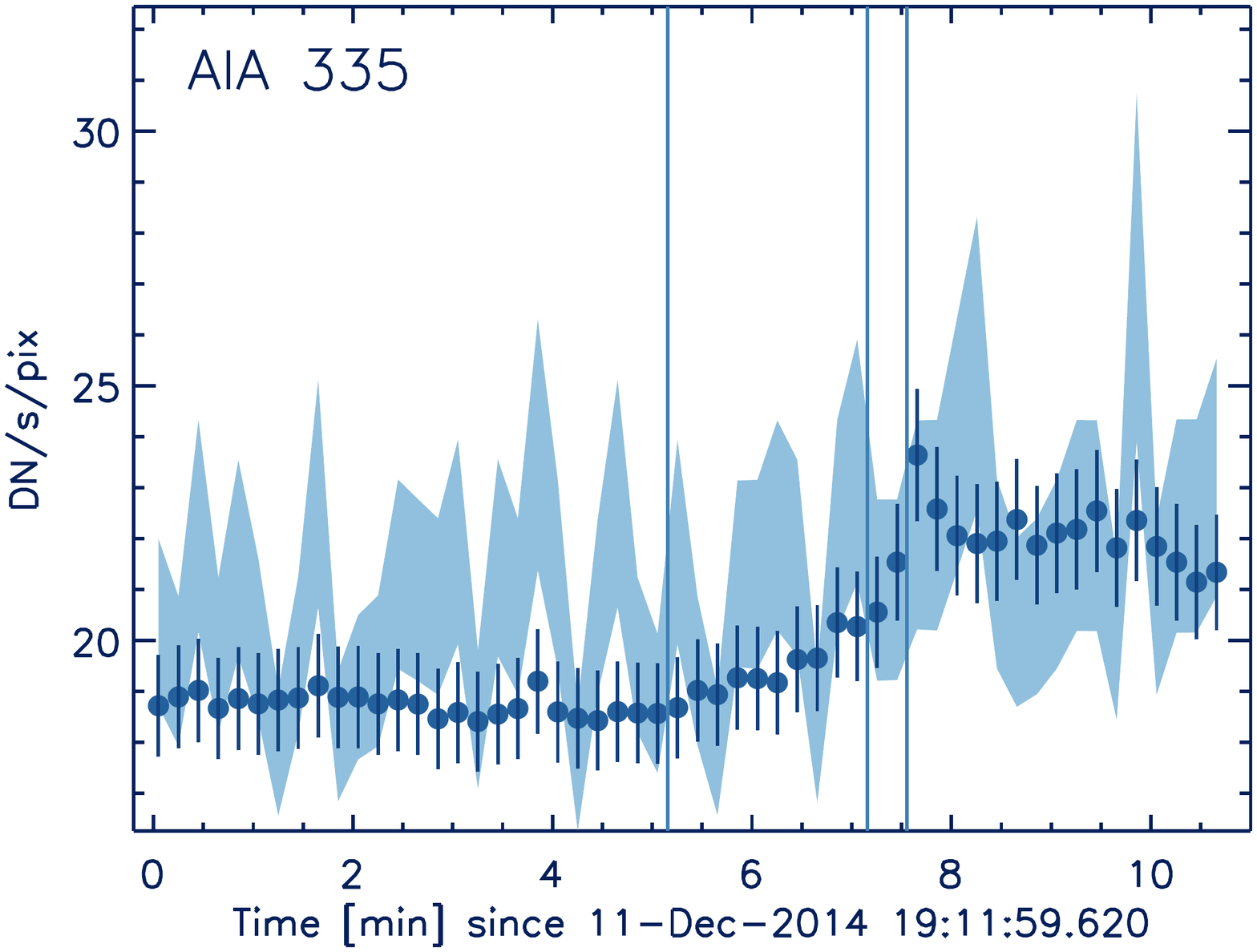}
    \includegraphics[width=3.3in]{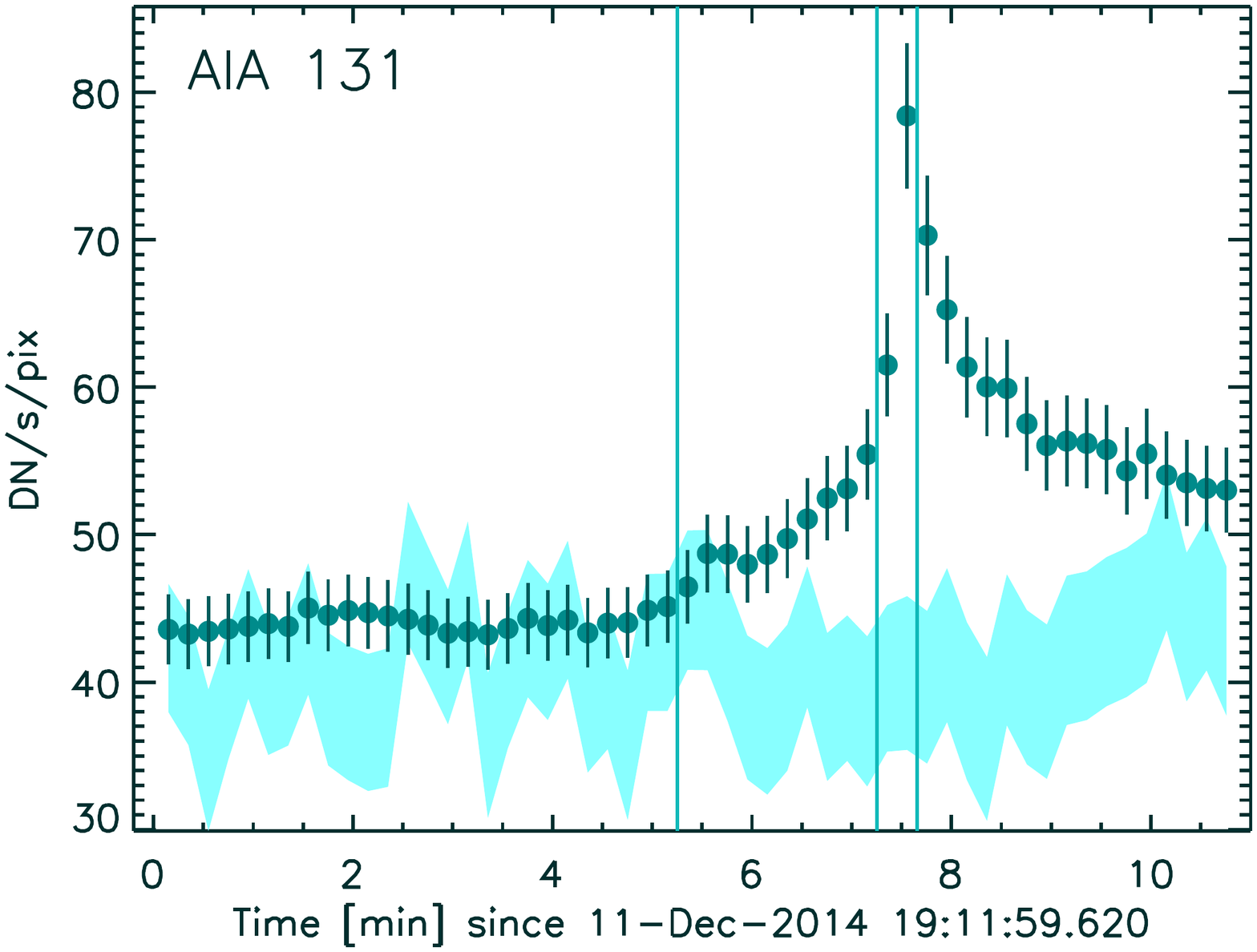}
    \caption{Light curves of the key pixels where changes are found, for the three filters used in the analysis: 94\,\AA\ ({\sl top}), 335\,\AA\ ({\sl middle}), and 131\,\AA\ ({\sl bottom}).  The average of the observed intensities, weighted by the number of times each pixel is flagged as a key pixel, are shown as dots, along with the similarly weighted sample standard deviation as vertical bars.  The shaded regions represent the envelope of the sample standard deviation seen {\sl outside} the flagged pixels.  The vertical lines denote the change points found by our algorithm.}
    \label{fig:solar_keypix_lightcurves}
\end{figure}

The fitted images for the 3-band case can be found in Figure~\ref{fig:Solar_fitted_3band}. Notice that there is a loop-shaped object that is of interest. Based on the fitted result, this object starts to appear at time point c.36 and becomes brighter after that for the first band. In the second band, this object appears at time point c.26 and stays bright throughout the duration considered. However in the third band, the object becomes bright at time point c.36 and vanishes soon after time point c.38. The proposed method is able to catch these changes in different bands and to detect the corresponding change points.

After detecting these change points, we find the key pixels that contribute to the change points using the methods in Section~\ref{subsubsection:highlight_difference}.
This method is appropriate to highlight the regions that change rapidly after the change point because different bands may not change in the same direction for this dataset. Here we apply this method on a single band, 94, as an example. We use $p = 10^{-15}$. See Figure~\ref{fig:Solar_highighted_method3_4.1} for an illustration. We find that the method could highlight the loop-shape object which starts to appear at time point c.36, and also detect the region that becomes much brighter after time point c.38. We also compute the light curves of the intensities in a region comprised of the set of pixels formed from the union of all key pixels found at all change points in all the filters; see Figure~\ref{fig:solar_keypix_lightcurves}.  Notice that the event of interest is fully incorporated within the key pixels, with no spillover into the background, and the change points are post facto found to be reasonably located from a temporal perspective in that they are located where a researcher seeking to manually place them would do so.  The first segment is characterized by steady emission in all three bands, the second segment shows the isolated loop beginning to form, the third segment catches the time when it reaches a peak, and the last segment tracks the slow decline in intensity.

\section{Summary}

We have developed an approach to model photon emissions by astronomical sources. Also, we propose a practical algorithm to detect the change points as well as to segment the astronomical images, based on the MDL principle for model selection. 
We test this method on a series of simulation experiments and apply it to two real astrophysical datasets. We are able to recover the time-evolving variations. 

Based on the results of simulation experiments, it is recommended that the average number of photon counts within each bin should be from 100 to 1000 for pixels belonging to an astrophysical object, so that the proposed method is able to find change points and limit false positives.

For future work, it will be helpful to quantify the evidence of the existence of a change point by deriving a test statistic based on Monte Carlo \rev{simulations} or other methods. Another possible extension is to relax the piecewise constant assumption and allow piecewise linear/quadratic modeling so that the method is able to capture more complicated and realistic patterns.

\acknowledgements
{Support for CX and TCML was partially provided by the National Science Foundation under grants DMS-1811405, DMS-1811661 and DMS-1916125.
Support for HMG was provided by the National Aeronautics and Space Administration through the Smithsonian Astrophysical Observatory contract SV3-73016 to MIT for Support of the Chandra X-Ray Center, which is operated by the Smithsonian Astrophysical Observatory for and on behalf of the National Aeronautics Space Administration under contract NAS8-03060.
VLK acknowledges support from NASA Contract NAS8-03060 to the Chandra X-ray Center. VLK thanks Durgesh Tripathi for the solar dataset used in this analysis.}

\facilities{XMM-Newton (EPIC), SDO (AIA)}

\software{CIAO \citep{2006SPIE.6270E..1VF}.
    astropy \citep{2013A&A...558A..33A,2018AJ....156..123A}.
    IDL \url{https://www.l3harrisgeospatial.com/Software-Technology/IDL}.
    SolarSoft \citep{2012ascl.soft08013F}.
    Automark \citep{wong2016}.
    \rev{4D\_Automark, a Python package for the proposed method can be downloaded from \url{https://github.com/kevinxucong/4D_Automark}. (doi:10.5281/zenodo.4451396)}
}

\clearpage
\appendix

\section{Statistical Consistency}
\label{appendix:A}

Here we prove that the MDL scheme is statistically consistent (see Section~\ref{sec:consistency}), thereby ensuring that the estimates of region segmentations and the Poisson intensities are reliable measures of the data. In the following, we assume that the size of the time bins $\Delta T_{t}=1, \forall 1 \leq t \leq \nt$, as $\nt$ increases to infinity. That is to say, first, we study that as these underlying nonhomogeneous Poisson processes are extending at the same rate, the size of the bins keeps fixed, which leads to increasing number of independent observations for any given part of this Poisson process. And by keeping the size of the bins fixed, we get rid of the case that the Poisson parameters keep varying as $\nt$ increases. Second, by setting $\Delta T_{t}=1$, the Poisson parameter is numerically equal to the Poisson rate, which ease the arguments. The proof can be extended if we relax this assumption. 

Given the above assumption, the photon counts have the following Poisson model,
\begin{equation}
    y_{i,t,w} \stackrel{i.i.d.}{\sim} \text{Poisson}(\lambda_{i,t,w}).
\end{equation}

Given change points $\bm{\tau} = (\tau_1, \tau_{2},..., \tau_{K})$, we set $\tau_{0}=0$ and $\tau_{K+1} = \nt$, and let $\nu_{k} = \tau_{k} / \nt, k=0,1,...\nt$ to be the normalized change points. The consistency results are based on $\nu_{k}$'s being fixed as $\nt$ increases.

The observed images within the same interval follow the same corresponding piecewise constant model given change points $\bm{\tau}$. Let $x_{{i,t-\tau_{k-1},w}}^{(k)} = y_{i,t,w}, \tau_{k-1}+1 \leq t \leq \tau_k$, $\bm{X}_{t-\tau_{k-1}}^{(k)} = \bm{Y}_{t} = \{y_{i,t,w} | i=1,...,\nx, w=1,...,\nw \}, \tau_{k-1}+1 \leq t \leq \tau_k$, and $\bm{X}^{(k)} = \{ \bm{X}_{t}^{(k)} | 1 \leq t \leq T_{k} \}$ where $T_{k} = \tau_{k} - \tau_{k-1}$.

Let $\lambda_{i,t,w}$'s within the $k^{\rm th}$ interval follow the same corresponding two-dimensional piecewise constant model with $m^{(k)}$ constant regions. The partition of the images within interval $k$ is specified by a region assignment function $R^{(k)}(.) \colon \{1,...,\nx \} \to \{ 1,...,m^{(k)} \}$. The Poisson parameters $\mu_{h,w}^{(k)}$ is the value for the $w^{\rm th}$ band in the $h^{\rm th}$ region of the $k^{\rm th}$ interval. Let $\mu^{(k)} = \{\mu_{h,w}^{(k)} | h=1,...,m^{(k)}, w = 1,...,\nw \}$. That is to say,
\begin{equation}
    \lambda_{i,t,w} =  \sum_{k=1}^{K+1} I_{ \{ t \in (\tau_{k-1},\tau_{k}] \} } \mu_{R^{(k)}(i), w}^{(k)} .
\end{equation}

In all, for pixel $i \in R_{h}^{(k)}$,
\begin{equation}
    x_{i,t,w}^{(k)} \stackrel{i.i.d.}{\sim} \text{Poisson}(\mu_{h,w}^{(k)}) .
\end{equation}

For each $1 \leq t \leq T_{k}$, the log-likelihood function for regions assignment $R^{(k)}$ and Poisson parameters $\mu^{(k)}$ is
\begin{equation}
    \tilde{l}_k ((R^{(k)}, \mu^{(k)}) ; \bm{X}_{t}^{(k)}) = \sum_{w=1}^{\nw}  \sum_{h=1}^{m^{(k)}} \sum_{i, s.t. R^{(k)}(i) = h} [x_{{i,t,w}}^{(k)} \log(\mu_{h,w}^{(k)}) - \mu_{h,w}^{(k)} - \log(x_{{i,t,w}}^{(k)} !)].
\end{equation}

As some of the terms in the log-likelihood function have nothing to do with the parameters to estimate, we remove these terms and write down the log-likelihood to be
\begin{equation}
    l_k ((R^{(k)}, \mu^{(k)}) ; \bm{X}_{t}^{(k)}) = \sum_{w=1}^{\nw}  \sum_{h=1}^{m^{(k)}} \sum_{i, s.t. R^{(k)}(i) = h} [x_{i,t,w}^{(k)} \log(\mu_{h,w}^{(k)}) - \mu_{h,w}^{(k)}].
\end{equation}

Define $\psi_{k} = (R^{(k)}, \mu^{(k)})$ to be the parameter set for the $k^{\rm th}$ interval, and $\bm{\mathcal{M}}$ to be the class of models $\psi_{k}$ can take value from. Then the log-likelihood for the $k^{\rm th}$ interval can be written as

\begin{equation}
    L_{T}^{(k)} (\psi_{k}; \bm{X}^{(k)}) = \sum_{t = 1}^{T_{k}} l_k ((R^{(k)}, \mu^{(k)}) ; \bm{X}_{t}^{(k)}).
\end{equation}

Let $\bm{\nu} = (\nu_{1},...,\nu_{K})$ be the normalized change point location vector, and $\bm{\psi} = (\psi_{1},...,\psi_{K+1})$ be the parameter vector. Then vector $(K, \bm{\nu}, \bm{\psi})$ can specify a model for this sequence of images. The MDL is derived to be

\begin{eqnarray}
\label{eqn:mdl}
    \text{MDL}(K, \bm{\nu}, \bm{\psi}) =&& K \log(\nt) + \sum_{k=1}^{K+1} [ m^{(k)} \log(\nx) + \frac{\log(3)}{2} \sum_{h=1}^{m^{(k)}} b_{h}^{(k)} + \frac{\nw}{2} \sum_{h=1}^{m^{(k)}} \log(([T \nu_{k}] - [T \nu_{k-1}] + 1) a_h^{(k)})] \nonumber \\
    &-& \sum_{k=1}^{K+1} L_{T}^{(k)} (\psi_{k}; \bm{X}^{(k)}).
\end{eqnarray}
Here the \rev{``}area" (number of pixels) and \rev{``}perimeter" (number of pixel edges) of region $R_{h}^{(k)}$ are denoted by $a_{h}^{(k)}$ and $b_{h}^{(k)}$.

In order to make sure that the change points are identifiable, we assume that there exists a $\epsilon_{\nu} > 0$ such that $\min_{1 \leq k \leq K+1} | \nu_{k} - \nu_{k-1} | > \epsilon_{\nu}$. Therefore, the number of change points is bounded by $K \leq [1/\epsilon_{\nu}] + 1$. And there exists a constraint $A_{\epsilon_{\nu}}^{K}$ of $\bm{\nu}$ where

\begin{equation}
    A_{\epsilon_{\nu}}^{K} = \{ \bm{\nu} \in (0,1)^{K} | 0 < \nu_{1} <...<\nu_{K} < 1, \nu_{k} - \nu_{k-1}  > \epsilon_{\nu}, \forall 1 \leq k \leq K+1 \} .
\end{equation}

Then the estimation of the model based on MDL is given by

\begin{equation}
\label{eqn:MDL_based}
    (\hat{K}_{T}, \hat{\bm{\nu}}_{T}, \hat{\bm{\psi}}_{T}) = \text{arg} \min_{K \leq [1/\epsilon_{\nu}] + 1, \bm{\nu} \in A_{\epsilon_{\nu}}^{K}, \bm{\psi} \in \bm{\mathcal{M}}} \frac{1}{\nt} \text{MDL}(K, \bm{\nu}, \bm{\psi}) .
\end{equation}
Here $\text{MDL}(K, \bm{\nu}, \bm{\psi})$ is defined in~(\ref{eqn:mdl}), $\hat{\bm{\nu}}_{T} = (\hat{\nu}_{1}, ..., \hat{\nu}_{\hat{K}})$ and $\hat{\bm{\psi}}_{T} = (\hat{\psi}_{1},...,\hat{\psi}_{\hat{K}+1})$, where $\hat{\psi}_{k} = (\hat{R}^{(k)}, \hat{\mu}^{(k)})$. And $\hat{\mu}^{(k)}$ is defined as 
\begin{equation}
    \hat{\mu}^{(k)} = \text{arg} \max_{\mu^{(k)} \in \Theta_{k}(\hat{R}^{(k)})} L_{T}^{(k)} ((\hat{R}^{(k)}, \mu^{(k)}); \hat{\bm{X}}^{(k)})
\end{equation}
with $\hat{\bm{X}}^{(k)} = \{ \bm{Y}_{t} | [T \hat{\nu}_{k-1}] < t \leq [T \hat{\nu}_{k}] \}$ denotes the estimated $k^{\rm th}$ interval of the sequence of images.

We further define the log-likelihood formed by a portion of the observations in the $k^{\rm th}$ interval by
\begin{equation}
    L_{T}^{(k)} (\psi_{k}, \nu_{d}, \nu_{u}; \bm{X}^{(k)}) = \sum_{t = [T_{k} \nu_{d}]+1}^{[T_{k} \nu_{u}]} l_k ((R^{(k)}, \mu^{(k)}) ; \bm{X}_{t}^{(k)}) ,
\end{equation}
where $0 \leq \nu_{d} < \nu_{u} \leq 1$ and $\nu_{u} - \nu_{d} > \epsilon_{\nu}$.

We denote 
\begin{equation}
    \sup_{\nu_{d}, \nu_{u}} \coloneqq \sup_{0 \leq \nu_{d} < \nu_{u} \leq 1, \nu_{u} - \nu_{d} > \epsilon_{\nu}}
\end{equation}
to simplify the notation.

In this setting, an extension need to be made such that $\nu_{d}$ and $\nu_{u}$ can be slightly outside $[0,1]$. It means that the $k^{\rm th}$ estimated interval could cover a part of the observations that belong to the $(k-1)^{\rm th}$ and $(k+1)^{\rm th}$ true intervals. Based on the formula (3.4) in \citet{davis_yau_2013}, for a real-value function $f_T(\nu_{d}, \nu_{u})$ on $\mathcal{R}^2$,

\begin{equation}
\label{eqn:supremum}
    \sup_{\underline{\nu_{d}}, \overline{\nu_{u}}} f_T(\nu_{d}, \nu_{u}) \xrightarrow{a.s.} 0
\end{equation}
is used to denote
\begin{equation}
    \sup_{-h_{T} < \nu_{d} < \nu_{u} < 1+r_{T}, \nu_{u} - \nu_{d} > \epsilon_{\nu}} f_T(\nu_{d}, \nu_{u}) \xrightarrow{a.s.} 0
\end{equation}
for any pre-specified positive-valued sequences $h_{T}$ and $r_{T}$, which cover to 0 as $\nt \rightarrow \infty$.


The following assumptions on true Poisson parameters $\mu_{h,w}^{o(k)}, 1 \leq h \leq m^{(k)}, 1 \leq w \leq \nw, 1\leq k \leq (K^{o}+1)$ are necessary.

\begin{assumption}
\label{assumption:1}
\begin{equation}
    0 < C_{d} \coloneqq \min_{k,h,w} \mu_{h,w}^{o(k)} \leq C_{u} \coloneqq \max_{k,h,w} \mu_{h,w}^{o(k)} < \infty .
\end{equation}
\end{assumption}

\begin{assumption}
\label{assumption:2}
For two true neighboring regions $R_{p}^{(k)}$ and $R_{q}^{(k)}$ at the $k^{\rm th}$ interval,
\begin{equation}
    \delta_{1} \coloneqq \min_{k,p,q,w} | \mu_{p,w}^{o(k)} - \mu_{q,w}^{o(k)} | > 0 .
\end{equation}
\end{assumption}

\begin{assumption}
\label{assumption:3}
For any two neighboring intervals $(k-1)$ and $k$
\begin{equation}
    \delta_{2} \coloneqq \min_{k} \max_{i, w} | \mu_{R^{(k)}(i), w}^{o(k)} - \mu_{R^{(k-1)}(i), w}^{o(k-1)} | > 0 .
\end{equation}
\end{assumption}


\begin{proposition}[v]
\label{prop_1}
For $k=1,...,K+1$ and any fixed $R^{(k)}$, there exists a $\epsilon>0$ such that,
\begin{equation}
\begin{split}
    \sup_{\mu^{(k)} \in \Theta_{k}(R^{(k)})} E | l_k ((R^{(k)}, \mu^{(k)}) ; \bm{X}_{t}^{(k)}) |^{v+\epsilon} < \infty , \\
    \sup_{\mu^{(k)} \in \Theta_{k}(R^{(k)})} E | l_k' ((R^{(k)}, \mu^{(k)}) ; \bm{X}_{t}^{(k)}) |^{v+\epsilon} < \infty , \\
    \sup_{\mu^{(k)} \in \Theta_{k}(R^{(k)})} E | l_k'' ((R^{(k)}, \mu^{(k)}) ; \bm{X}_{t}^{(k)}) | < \infty .
\end{split}
\end{equation}
\end{proposition}

This proposition holds for $v=1,2,4$ due to the compactness of parameter space (Assumption \ref{assumption:1}) and bounded $E[(x_{i,t,w}^{(k)})^{v+\epsilon}]$.

\begin{proposition}
For $k=1,...,K+1$ and any fixed $R^{(k)}$,
\begin{equation}
\begin{split}
    \sup_{\mu^{(k)} \in \Theta_{k}(R^{(k)})} | \frac{1}{\nt (\nu_{k} - \nu_{k-1})} L_{T}^{(k)} ((R^{(k)}, \mu^{(k)}); \bm{X}^{(k)}) - L_{k}((R^{(k)}, \mu^{(k)})) | \xrightarrow{a.s.} 0 ,\\
    \sup_{\mu^{(k)} \in \Theta_{k}(R^{(k)})} | \frac{1}{\nt (\nu_{k} - \nu_{k-1})} L_{T}'^{(k)} ((R^{(k)}, \mu^{(k)}); \bm{X}^{(k)}) - L_{k}'((R^{(k)}, \mu^{(k)})) | \xrightarrow{a.s.} 0 ,\\
    \sup_{\mu^{(k)} \in \Theta_{k}(R^{(k)})} | \frac{1}{\nt (\nu_{k} - \nu_{k-1})} L_{T}''^{(k)} ((R^{(k)}, \mu^{(k)}); \bm{X}^{(k)}) - L_{k}''((R^{(k)}, \mu^{(k)})) | \xrightarrow{a.s.} 0 ,
\end{split}
\end{equation}

where
\begin{equation}
\begin{split}
    L_{k}((R^{(k)}, \mu^{(k)})) \coloneqq E(l_k ((R^{(k)}, \mu^{(k)}) ; \bm{X}_{t}^{(k)})) ,\\
    L_{k}'((R^{(k)}, \mu^{(k)})) \coloneqq E(l_k' ((R^{(k)}, \mu^{(k)}) ; \bm{X}_{t}^{(k)})) ,\\
    L_{k}''((R^{(k)}, \mu^{(k)})) \coloneqq E(l_k'' ((R^{(k)}, \mu^{(k)}) ; \bm{X}_{t}^{(k)})) .
\end{split}
\end{equation}
\end{proposition}

The estimated locations of change points are used to define the likelihood in practice. Therefore, the two ends of the $k^{\rm th}$ interval might contain observations from the $(k-1)^{\rm th}$ and $(k+1)^{\rm th}$ true intervals, though the estimated change points are close to the true change points. It is necessary to control the effect at the two ends of the fitted interval.

\begin{proposition}[w]
\label{prop_3}
For $k=1,...,K+1$ and any fixed $\psi$ and any sequence of integers $\{ g(\nt) \}_{\nt \geq 1} $ that satisfies $g(\nt) > c \nt^{w}$ for some $c>0$ when $\nt$ is large enough, then

\begin{equation}
\begin{split}
    \frac{1}{g(\nt)} \sum_{t=\nt-g(\nt)+1}^{\nt} l_k (\psi ; \bm{X}_{t}^{(k)}) \xrightarrow{a.s.} E(l_k (\psi ; \bm{X}_{t}^{(k)})) ,\\
    \frac{1}{g(\nt)} \sum_{t=\nt-g(\nt)+1}^{\nt} l_k' (\psi ; \bm{X}_{t}^{(k)}) \xrightarrow{a.s.} E(l_k' (\psi ; \bm{X}_{t}^{(k)})) .
\end{split}
\end{equation}
\end{proposition}

Based on Lemma 1 in \citet{davis_yau_2013}, Proposition \ref{prop_3} holds when Proposition \ref{prop_1}(2) holds and the Assumption 4* in \citet{davis_yau_2013} is satisfied. And Assumption 4* is satisfied because \rev{an} independent process, like the current setting, must be mixing.

It is necessary to discuss the identifiability of models in $\bm{\mathcal{M}}$. First we define $R^{b}$ an oversegmentation compared with $R^{s}$ if $R^{b}(i) = R^{b}(j)$ leads to $R^{s}(i) = R^{s}(j)$. 

\begin{proposition}
\label{prop_4}
For the $k^{\rm th}$ interval, the true model $\psi_{k}^{o} \in \bm{\mathcal{M}}$ satisfies $\psi_{k}^{o} = \arg \max_{\psi \in \bm{\mathcal{M}}} E ( l_k (\psi ; \bm{X}_{t}^{(k)})) $. Also, $\psi_{k}^{o}$ is uniquely identifiable, which means that if there exists a $\mu^{*}$ such that $l_k ((R^{o}, \mu^{o}) ; \bm{X}_{t}^{(k)}) = l_k ((R^{o}, \mu^{*}) ; \bm{X}_{t}^{(k)})$ almost everywhere for $\bm{X}_{t}^{(k)}$, then $\mu^{o} = \mu^{*}$. And suppose there exists another model $\psi_{k}^{b} = (R^{b}, \mu^{b})$ such that $l_k (\psi_{k}^{b} ; \bm{X}_{t}^{(k)}) = l_k (\psi_{k}^{o} ; \bm{X}_{t}^{(k)}) $ almost everywhere, then $R^{b}$ must be an oversegmentation compared with $R^{o}$. And $\mu^{b}$ satisfies  $\mu_{R^{b}(i), w}^{b} = \mu_{R^{o}(i), w}^{o}, \forall i,w$.
\end{proposition}

\begin{proof}
Suppose on the contrary there exist a model $\psi^{*} = (R^{*}, \mu^{*})$ that satisfies $\psi^{*} = \arg \max_{\psi \in \bm{\mathcal{M}}} E ( l_k (\psi ; \bm{X}_{t}^{(k)})) ) $, and $\psi^{*}$ is neither the true model nor an oversegmentation of the true model. Then there exist two pixels $i_0$ and $j_0$, such that they are neighboring pixels, $R^{o(k)}(i_0) \neq R^{o(k)}(j_0) \text{ and } R^{*}(i_0) = R^{*}(j_0)$. Therefore, by Assumption \ref{assumption:2}, we have
\begin{equation}
\label{eqn:prop4_assumption}
    |\mu_{R^{o(k)}(i_0),w}^{o(k)}-\mu_{R^{o(k)}(j_0),w}^{o(k)}| \geq \delta_1 > 0 .
\end{equation}

Define
\begin{equation}
    \bar{\mu}_{h, w}(R) \coloneqq \frac{1}{a_{h}(R)} \sum_{i, R(i)=h} \mu_{R^{o}(i), w}^{o(k)} ,
\end{equation}
where $a_{h}(R)$ denotes the number of pixels in region $h$ given segmentation $R$.
And in a special case,
\begin{equation}
    \bar{\mu}_{h, w}(R^{o(k)}) = \mu_{h,w}^{o(k)} .
\end{equation}

Then for all possible $\mu^{*} \in \Theta_{k}(R^{*})$, we have
\begin{equation}
\label{eqn:prop4_1}
\begin{split}
    E ( l_k ((R^{*}, \mu^{*}) ; \bm{X}_{t}^{(k)}) ) &= E (\sum_{w=1}^{\nw}  \sum_{h=1}^{m^{*}} \sum_{i, s.t. R^{*}(i) = h} [x_{i,t,w}^{(k)} \log(\mu_{h,w}^{*}) - \mu_{h,w}^{*}] )\\
    &= \sum_{w=1}^{\nw} \sum_{h=1}^{m^{*}} \sum_{i, s.t. R^{*}(i) = h} [\mu_{R^{o(k)}(i),w}^{o(k)} \log(\mu_{h,w}^{*}) - \mu_{h,w}^{*}]\\
    &= \sum_{w=1}^{\nw} \sum_{h=1}^{m^{*}} a_{h}(R^{*}) [\bar{\mu}_{h, w}(R) \log(\mu_{h,w}^{*}) - \mu_{h,w}^{*}]\\
    &\leq \sum_{w=1}^{\nw} \sum_{h=1}^{m^{*}} a_{h}(R^{*}) \max_{\mu_{h,w}} [\bar{\mu}_{h, w}(R) \log(\mu_{h,w}) - \mu_{h,w}]\\
    &= \sum_{w=1}^{\nw} \sum_{h=1}^{m^{*}} a_{h}(R^{*}) [\bar{\mu}_{h, w}(R^{*}) \log(\bar{\mu}_{h, w}(R^{*})) - \bar{\mu}_{h, w}(R^{*})]\\
    &= E ( l_k ((R^{*}, \bar{\mu}(R^{*})) ; \bm{X}_{t}^{(k)}) ) .
\end{split}
\end{equation}

Also, we have
\begin{equation}
\label{eqn:prop4_2}
\begin{split}
    E ( l_k ((R^{*}, \bar{\mu}(R^{*})) ; \bm{X}_{t}^{(k)}) ) =& 
    \sum_{w=1}^{\nw} \sum_{h=1}^{m^{*}} a_{h}(R^{*}) \max_{\mu_{h,w}} [\bar{\mu}_{h, w}(R^{*}) \log(\mu_{h, w}) - \mu_{h, w}]\\
    \leq & \sum_{w=1}^{\nw} \sum_{i=1, i \notin \{i_0, j_0\} }^{\nx} \max_{\lambda_{i,w}} [ \mu_{R^{o(k)}(i),w}^{o(k)} \log(\lambda_{i,w}) - \lambda_{i,w}] \\
    & + \sum_{w=1}^{\nw} [ \mu_{R^{o(k)}(i_0),w}^{o(k)} \log(\mu_{R^{*}(i_0),w}^{*}) - \mu_{R^{*}(i_0),w}^{*} + \mu_{R^{o(k)}(j_0),w}^{o(k)} \log(\mu_{R^{*}(j_0),w}^{*}) - \mu_{R^{*}(j_0),w}^{*}] \\
    < & \sum_{w=1}^{\nw} \sum_{i=1, i \notin \{i_0, j_0\} }^{\nx} \max_{\lambda_{i,w}} [ \mu_{R^{o(k)}(i),w}^{o(k)} \log(\lambda_{i,w}) - \lambda_{i,w}] \\
    & + \sum_{w=1}^{\nw} \max_{\lambda_{i_0,w}} [ \mu_{R^{o(k)}(i_0),w}^{o(k)} \log(\lambda_{i_0,w}) - \lambda_{i_0,w}]\\
    & + \sum_{w=1}^{\nw} \max_{\lambda_{j_0,w}} [ \mu_{R^{o(k)}(j_0),w}^{o(k)} \log(\lambda_{j_0,w}) - \lambda_{j_0,w}]\\
    = & \sum_{w=1}^{\nw} \sum_{i=1}^{\nx} [ \mu_{R^{o(k)}(i),w}^{o(k)} \log(\mu_{R^{o(k)}(i),w}^{o(k)}) - \mu_{R^{o(k)}(i),w}^{o(k)}]\\
    = & \sum_{w=1}^{\nw} \sum_{h=1}^{m^{o(k)}} a_{h}^{o(k)}[ \mu_{R^{o(k)}(i),w}^{o(k)} \log(\mu_{R^{o(k)}(i),w}^{o(k)}) - \mu_{R^{o(k)}(i),w}^{o(k)}]\\
    = & E ( l_k ((R^{o(k)}, \mu^{o(k)}) ; \bm{X}_{t}^{(k)}) ) .
\end{split}
\end{equation}
Here the strict inequities must hold because of~(\ref{eqn:prop4_assumption})

Finally combining~(\ref{eqn:prop4_1}) and~(\ref{eqn:prop4_2}), we have
\begin{equation}
    E ( l_k ((R^{*}, \mu^{*}) ; \bm{X}_{t}^{(k)}) ) < E ( l_k ((R^{o(k)}, \mu^{o(k)}) ; \bm{X}_{t}^{(k)}) ) ,
\end{equation}
which is a contradiction. This finishes the proof.
\end{proof}

\begin{lemma}
\label{lemma:1}
For any fixed $R^{(k)}$,
\begin{equation}
\begin{split}
    \sup_{\underline{\nu_{d}}, \overline{\nu_{u}}} \sup_{\mu^{(k)} \in \Theta_{k}(R^{(k)})} | \frac{1}{\nt (\nu_{k} - \nu_{k-1})} L_{T}^{(k)} ( (R^{(k)}, \mu^{(k)}) , \nu_{d}, \nu_{u}; \bm{X}^{(k)}) - (\nu_{u} - \nu_{d}) L_{k}((R^{(k)}, \mu^{(k)})) | \xrightarrow{a.s.} 0 ,\\
    \sup_{\underline{\nu_{d}}, \overline{\nu_{u}}} \sup_{\mu^{(k)} \in \Theta_{k}(R^{(k)})} | \frac{1}{\nt (\nu_{k} - \nu_{k-1})} L_{T}'^{(k)} ( (R^{(k)}, \mu^{(k)}) , \nu_{d}, \nu_{u}; \bm{X}^{(k)}) - (\nu_{u} - \nu_{d}) L_{k}'((R^{(k)}, \mu^{(k)})) | \xrightarrow{a.s.} 0 ,\\
    \sup_{\underline{\nu_{d}}, \overline{\nu_{u}}} \sup_{\mu^{(k)} \in \Theta_{k}(R^{(k)})} | \frac{1}{\nt (\nu_{k} - \nu_{k-1})} L_{T}''^{(k)} ( (R^{(k)}, \mu^{(k)}) , \nu_{d}, \nu_{u}; \bm{X}^{(k)}) - (\nu_{u} - \nu_{d}) L_{k}''((R^{(k)}, \mu^{(k)})) | \xrightarrow{a.s.} 0 .
\end{split}
\end{equation}

\end{lemma}

See Proposition 1 and 2 in \citet{davis_yau_2013} for the proof.

\begin{lemma}
Suppose the true parameters for interval $k$ is $\psi^{o(k)} = (R^{o(k)}, \mu^{o(k)})$. And suppose a region segmentation $R^{(k)}$ is specified for estimation. Let

\begin{equation}
\begin{split}
    \hat{\mu}_{T} &= \hat{\mu}_{T}^{(k)}(\nu_{d}, \nu_{u}) \coloneqq \arg \max_{\mu^{(k)} \in \Theta_{k}(R^{(k)})} L_{T}^{(k)} ( (R^{(k)}, \mu^{(k)}) , \nu_{d}, \nu_{u}; \bm{X}_{k}) ,\\
    \mu^{*(k)} &\coloneqq \arg \max_{\mu^{(k)} \in \Theta_{k}(R^{(k)})} L_{k}((R^{(k)}, \mu^{(k)})) .
\end{split}
\end{equation}

Then

\begin{equation}
\label{eqn:prop_6.1}
    \sup_{\underline{\nu_{d}}, \overline{\nu_{u}}} | \frac{1}{\nt (\nu_{k} - \nu_{k-1})} L_{T}^{(k)} ( (R^{(k)}, \hat{\mu}_{T}) , \nu_{d}, \nu_{u}; \bm{X}^{(k)}) - (\nu_{u} - \nu_{d}) L_{k}((R^{(k)}, \mu^{*(k)})) | \xrightarrow{a.s.} 0 ,
\end{equation}
where the supremum is defined in (\ref{eqn:supremum}). And if $R^{(k)} = R^{o(k)}$, we further have

\begin{equation}
\label{eqn:prop_6.2}
    \sup_{\underline{\nu_{d}}, \overline{\nu_{u}}} | \hat{\mu}_{T}^{(k)}(\nu_{d}, \nu_{u}) - \mu^{o(k)} | \xrightarrow{a.s.} 0 .
\end{equation}

If $R^{(k)}$ is an oversegmentation than $R^{o(k)}$, then we have

\begin{equation}
\label{eqn:prop_6.3}
    \sup_{\underline{\nu_{d}}, \overline{\nu_{u}}} | \hat{\mu}_{T, R^{(k)}(i), w}(\nu_{d}, \nu_{u}) - \mu^{o(k)}_{R^{o(k)}(i), w} | \xrightarrow{a.s.} 0 \ \forall i,w .
\end{equation}

\end{lemma}

\begin{proof}
\begin{equation}
\label{eqn:prop6.4}
\begin{split}
    &(\nu_{u} - \nu_{d}) (L_{k}((R^{(k)}, \mu^{*(k)})) - L_{k}((R^{(k)}, \hat{\mu}_{T}))) \\ 
    \leq& \sup_{\underline{\nu_{d}}, \overline{\nu_{u}}} | (\nu_{u} - \nu_{d}) L_{k}((R^{(k)}, \mu^{*(k)})) - \frac{1}{\nt (\nu_{k} - \nu_{k-1})} L_{T}^{(k)} ( (R^{(k)}, \mu^{*(k)}) , \nu_{d}, \nu_{u}; \bm{X}^{(k)}) \\
    &+ \frac{1}{\nt (\nu_{k} - \nu_{k-1})} L_{T}^{(k)} ( (R^{(k)}, \hat{\mu}_{T}) , \nu_{d}, \nu_{u}; \bm{X}^{(k)}) - (\nu_{u} - \nu_{d}) L_{k}((R^{(k)}, \hat{\mu}_{T})) | \\
    \leq& 2 \sup_{\underline{\nu_{d}}, \overline{\nu_{u}}} \sup_{\mu^{(k)} \in \Theta_{k}(R^{(k)})} | \frac{1}{\nt (\nu_{k} - \nu_{k-1})} L_{T}^{(k)} ( (R^{(k)}, \hat{\mu}_{T}) , \nu_{d}, \nu_{u}; \bm{X}^{(k)}) - (\nu_{u} - \nu_{d}) L_{k}((R^{(k)}, \mu^{(k)})) | \\ \xrightarrow{a.s.}& 0 .
\end{split}
\end{equation}

The first inequity is obtained by the definition of maximum likelihood estimator, and the last convergence comes from Lemma~{\ref{lemma:1}}. As $\mu^{*(k)}$ maximizes $L_{k}((R^{(k)}, \mu^{(k)}))$ and $\nu_{u} - \nu_{d} > 0$, we have

\begin{equation}
\label{eqn:prop6.5}
    | L_{k}((R^{(k)}, \mu^{*(k)})) - L_{k}((R^{(k)}, \hat{\mu}_{T})) | \xrightarrow{a.s.} 0 .
\end{equation}

Combining~(\ref{eqn:prop6.4}), (\ref{eqn:prop6.5}) and Proposition \ref{prop_1}(1), \ref{eqn:prop_6.1} holds. If $R^{(k)} = R^{o(k)}$, by Proposition \ref{prop_4}, $L_{k}((R^{(k)}, \mu^{(k)}))$ has a unique maximizer at $\mu^{o(k)}$, so~(\ref{eqn:prop_6.2}) holds. If $R^{(k)}$ is an oversegmentation \rev{compared with} $R^{o(k)}$, by Proposition \ref{prop_4}, (\ref{eqn:prop_6.3}) holds.
\end{proof}

Now we give a preliminary result of the convergence when the number of change points is known.

\begin{theorem} (Theorem 1 in \citet{davis_yau_2013})
\label{theorem_1}
Let $\{ \bm{Y}_{t} | t=1,..., \nt \}$ be the observed images specified by $(K^{o}, \bm{\nu}^{o}, \bm{\psi}^{o})$. And suppose the number of change points $K^{o}$ is known. The change points and parameters are estimated by

\begin{equation}
    (\hat{\bm{\nu}}_{T}, \hat{\bm{\psi}}_{T}) = \text{arg} \min_{\bm{\lambda} \in A_{\epsilon_{\lambda}}^{m}, \bm{\psi} \in \bm{\mathcal{M}}} \frac{1}{\nt} \text{MDL}(K^{o}, \bm{\nu}, \bm{\psi}) .
\end{equation}

Then $\hat{\bm{\nu}}_{T} \xrightarrow{a.s.} \bm{\nu}^{o}$ and for each interval, the estimated $\hat{R}^{(k)}$ must be an oversegmentation comparing to the true region segmentation.

\end{theorem}

We skip the proof of this theorem because it is quite similar to the proof of Theorem 1 in \citet{davis_yau_2013}. Notice that we need to use Assumption \ref{assumption:3} in the proof.

\begin{corollary} (Corollary 1 in \citet{davis_yau_2013})
Under the conditions of Theorem \ref{theorem_1}, if the number of change-
points is unknown and is estimated from the data , then
\begin{enumerate}
    \item The number of change points cannot be underestimated. That is to say, $\hat{K} \geq K^{o}$ almost surely when $\nt$ is large enough.
    \item When $\hat{K} > K^{o}$, $\bm{\nu}^{o}$ must be a subset of the limit of $\hat{\bm{\nu}}_{T}$ for large enough $\nt$.
    \item In each fitted interval, the region segmentation must be equal to or be an oversegmentation comparing with the corresponding true region segmentation. 
\end{enumerate}
\end{corollary}

See Corollary 1 in \citet{davis_yau_2013} for more details.

\begin{theorem} (Theorem 2 in \citet{davis_yau_2013})
Let $\bm{\nu}^{o} = (\nu_{1}^{o}, \nu_{2}^{o},...,\nu_{m^{o}}^{o})$ be the true change points. And $(\hat{K}, \hat{\bm{\nu}}_{T}, \hat{\bm{\psi}}_{T})$ is the MDL-based result. Then $\forall k=1,2,...,K^{o}$, there exists a $\hat{\nu}_{t_{k}} \in \hat{\bm{\nu}}_{T}$ where $1 \leq t_{k} \leq \hat{K}$ such that

\begin{equation}
    | \hat{\nu}_{t_{k}} - \nu_{k}^{o} | = o(\nt ^{-\frac{1}{2}}) \ a.s. \ .
\end{equation}
\end{theorem}

See the proof of Theorem 2 in \citet{davis_yau_2013}.

\begin{lemma}
Suppose the true region segmentation $R^{o(k)}$ is specified for the $k^{\rm th}$ interval, then

\begin{equation}
    \hat{\mu}_{T}^{(k)}(\hat{\nu}_{k-1}, \hat{\nu}_{k}) - \mu^{o(k)} = O(\sqrt{\frac{\log \log(\nt)}{\nt}}) \ a.s. \ .
\end{equation}

When the specific region segmentation $R^{(k)}$ is an oversegmentation \rev{compared with} $R^{o(k)}$, then we have

\begin{equation}
    \hat{\mu}_{T, R^{(k)}(i), w}^{(k)}(\hat{\nu}_{k-1}, \hat{\nu}_{k}) - \mu_{R^{o(k)}(i), w}^{o(k)} = O(\sqrt{\frac{\log \log(\nt)}{\nt}}) \ \ a.s.\ \ \forall i, w .
\end{equation}
\end{lemma}

See Lemma 2 in \citet{davis_yau_2013} for more details.

Then we come to the main result.

\begin{theorem}
Let $\{ \bm{Y}_{t} | t=1,...,\nt \}$ be the observed images specified by $(K^{o}, \bm{\nu}^{o}, \bm{\psi}^{o})$. The estimator $(\hat{K}_{T}, \hat{\bm{\nu}}_{T}, \hat{\bm{\psi}}_{T})$ is defined by~(\ref{eqn:MDL_based}). Then we have
\begin{equation}
\begin{split}
    \hat{K}_{T} \xrightarrow{a.s.} K^{o}, \\
    \hat{\bm{\nu}}_{T} \xrightarrow{a.s.} \bm{\nu}^{o}, \\
    \hat{\bm{\psi}}_{T} \xrightarrow{a.s.} \bm{\psi}^{o} .
\end{split}
\end{equation}

\end{theorem}

See Theorem 3 in \citet{davis_yau_2013} for more details.

\bibliography{ref_change_points}
\bibliographystyle{aasjournal}

\end{document}